\documentclass[preprint]{imsart}

\RequirePackage{amsthm,amsmath,amsfonts,amssymb}
\RequirePackage[numbers]{natbib}
\usepackage[a4paper, total={5in, 8in},centering]{geometry}
 \usepackage{siunitx}
 \usepackage{xargs,pdfsync}
\usepackage[shortlabels]{enumitem}

\usepackage{graphicx}
\usepackage{aliascnt}
\usepackage{amsthm}
\usepackage{bbm}
\usepackage{multirow}
\usepackage{framed}
\usepackage{hyperref}
\usepackage{stmaryrd}
\usepackage{xcolor}

\usepackage{cleveref}

\usepackage{algorithm}
\usepackage{algpseudocode}

\newenvironment{hyp}[1]{
\begin{enumerate}[label=\textbf{\sf(#1\arabic*)},resume=hyp#1]\begin{sf}}{\end{sf}\end{enumerate}}
\crefname{hyp}{}{ass}
\Crefname{hyp}{}{Ass}

\graphicspath{{../../Figures/}} 
\usepackage[disable]{todonotes}
\usepackage{upgreek}

\startlocaldefs

\newtheorem{theorem}{Theorem}[section]
\newaliascnt{proposition}{theorem}
\newtheorem{proposition}[proposition]{Proposition}
\aliascntresetthe{proposition}

\newaliascnt{lemma}{theorem}
\newtheorem{lemma}[lemma]{Lemma}
\aliascntresetthe{lemma}

\newaliascnt{corollary}{theorem}
\newtheorem{corollary}[corollary]{Corollary}
\aliascntresetthe{corollary}

\newaliascnt{definition}{theorem}

\aliascntresetthe{definition}

\newaliascnt{example}{theorem}

\aliascntresetthe{example}

\newaliascnt{remark}{theorem}
\newtheorem{remark}[remark]{Remark}
\aliascntresetthe{remark}

\newcommand{\Aset}{\mathsf{A}}

\newcommand{\1}{\mathbbm{1}}
\newcommand{\as}{\mathrm{a.s.}}
\newcommandx{\convps}[1][1=]{\displaystyle \stackrel[#1]{\PP\mathrm{-a.s.}}{\longrightarrow}}
\newcommand{\bddfunc}[1]{\mathsf{F}_{\mathrm b}(#1)}
\newcommand{\bCset}{ {\bar{\mathsf C}}}
\newcommand{\bDset}{ {\bar{\mathsf D}}}
\newcommand{\Bset}{\mathsf{B}}

\newcommandx{\convprob}[1][1=]{\displaystyle \stackrel[#1]{\PP}{\longrightarrow}}

\newcommand{\Cset}{\mathsf{C}}
\newcommand{\Csigma}{\mathcal{E}_\Cset}

\newcommandx{\dlim}[1][1=]{\displaystyle \stackrel[#1]{\PP}{\rightsquigarrow}}
\newcommand{\Dset}{ {\mathsf D}}

\newcommand{\eqdef}{:=}
\newcommand{\eqsp}{}

\newcommand{\Eset}{{\mathsf E}}
\newcommand{\Esigma}{\mathcal E}
\newcommand{\Fset}{{\mathsf F}}
\newcommand{\Fsigma}{{\mathcal F}}

\newcommand{\gKKT}{{GKKT}}

\newcommand{\indi}[1]{\1_{#1}}
\newcommand{\indiacc}[1]{\1_{\lrcb{#1}}}
\newcommand{\indin}[1]{\1_{\{#1\}}}

\newcommand{\invar}{\mathcal{I}}

\newcommand{\lr}[1]{\left(#1 \right)}

\newcommand{\lrb}[1]{\left[#1 \right]}
\newcommand{\lrcb}[1]{\{#1\}}
\newcommand{\MP}{\eta}
\newcommand{\measureset}{\mathsf{M}}

\newcommand{\nonnegfunc}[1]{\mathsf{F}_+(#1)}

\newcommand{\nset}{{\mathbb N}}
\newcommand{\nsetpos}{\mathbb{N}^\ast}
\newcommand{\PE}{\mathbb E}
\newcommand{\PEE}[2]{\mathbb E^{#1}_{#2}}

\newcommand{\posfunc}[1]{\mathsf{F}_+(#1)}
\newcommand{\PP}{\mathbb P}
\newcommand{\PPP}[2]{ {\mathbb P}^{#1}_{#2}}

\newcommand{\rmd}{\mathrm d}

\newcommand{\rset}{\mathbb{R}}
\newcommand{\rsetps}{\mathbb{R}^*_+}
\newcommand{\rti}{\sigma}
\newcommand{\rtiC}[1]{\sigma_\Cset^{#1}}

\newcommand{\seq}[2]{(#1)_{#2}}

\newcommand{\set}[2]{\lrcb{#1\,:\eqsp #2}}

\newcommand{\shift}{\theta}

\newcommand{\tensprod}{\varotimes}

\newcommand{\X}[2]{X_{#1}^{#2}}

\newcommand{\Xinit}{\mu}

\def\msa{\mathsf{A}}

\def\msc{\mathsf{C}}
\def\bmsc{\bar{\mathsf{C}}}
\def\bmsc{\bar{\mathsf{C}}}
\def\mse{\mathsf{E}}

\def\msl{\mathsf{L}}

\def\msx{\mathsf{X}}

\def\msd{\mathsf{D}}

\def\mcbb{\mathcal{B}}  

\def\mce{\mathcal{E}}

\def\mcf{\mathcal{F}}
\def\mcg{\mathcal{G}}

\def\rset{\mathbb{R}}

\def\nset{\mathbb{N}}
\def\nsets{{\mathbb{N}^*}}

\def\rmd{\mathrm{d}}

\def\rme{\mathrm{e}}

\newcommandx{\functionspace}[2][1=+]{\mathbb{F}_{#1}(#2)}

\newcommandx{\VarDeux}[3][3=]{\operatorname{Var}^{#3}_{#1}\left\{#2 \right\}}

\newcommand{\LeftEqNo}{\let\veqno\@@leqno}

\newcommand{\abs}[1]{\left\vert #1 \right\vert}

\newcommand{\tvnorm}[1]{\| #1 \|_{\mathrm{TV}}}

\newcommandx{\Vnorm}[2][1=V]{\| #2 \|_{#1}}
\newcommandx{\VnormEq}[2][1=V]{\left\| #2 \right\|_{#1}}
\newcommandx{\norm}[2][1=]{\ifthenelse{\equal{#1}{}}{\left\Vert #2 \right\Vert}{\left\Vert #2 \right\Vert^{#1}}}
\newcommandx{\normLigne}[2][1=]{\ifthenelse{\equal{#1}{}}{\Vert #2 \Vert}{\Vert #2\Vert^{#1}}}
\newcommandx{\normL}[2][1=]{\ifthenelse{\equal{#1}{}}{\left\Vert #2 \right\Vert_{\msl}}{\left\Vert #2 \right\Vert^{#1}_{\msl}}}
\newcommandx{\normLLigne}[2][1=]{\ifthenelse{\equal{#1}{}}{\Vert #2 \Vert_{\msl}}{\Vert #2\Vert^{#1}_{\msl}}}

\newcommand{\parenthese}[1]{\left(#1 \right)}

\newcommand{\parentheseDeux}[1]{\left[ #1 \right]}

\newcommand{\defEns}[1]{\left\lbrace #1 \right\rbrace }

\newcommand{\plusinfty}{\infty}

\def\ie{\textit{i.e.}}

\def\eqsp{\,}

\newcommand{\ooint}[1]{\left(#1\right)}
\newcommand{\ccint}[1]{\left[#1\right]}

\newcommandx{\weight}[2][2=n]{\omega_{#1,#2}^N}

\def\as{\ensuremath{\text{a.s.}}}

\newcommandx\sequence[3][2=,3=]
{\ifthenelse{\equal{#3}{}}{\ensuremath{\{ #1_{#2}\}}}{\ensuremath{\{ #1_{#2}, \eqsp #2 \in #3 \}}}}
\newcommandx\sequenceD[3][2=,3=]
{\ifthenelse{\equal{#3}{}}{\ensuremath{\{ #1_{#2}\}}}{\ensuremath{( #1)_{ #2 \in #3} }}}

\newcommandx{\sequencen}[2][2=n\in\N]{\ensuremath{\{ #1_n, \eqsp #2 \}}}
\newcommandx\sequenceDouble[4][3=,4=]
{\ifthenelse{\equal{#3}{}}{\ensuremath{\{ (#1_{#3},#2_{#3}) \}}}{\ensuremath{\{  (#1_{#3},#2_{#3}), \eqsp #3 \in #4 \}}}}
\newcommandx{\sequencenDouble}[3][3=n\in\N]{\ensuremath{\{ (#1_{n},#2_{n}), \eqsp #3 \}}}

\def\iid{i.i.d.}

\newcommand{\opnorm}[1]{{\left\vert\kern-0.25ex\left\vert\kern-0.25ex\left\vert #1
    \right\vert\kern-0.25ex\right\vert\kern-0.25ex\right\vert}}

\newcommandx{\CPE}[3][1=]{{\mathbb E}_{#1}\left[#2 \left \vert #3 \right. \right]} 
\newcommandx{\CPVar}[3][1=]{\mathrm{Var}^{#3}_{#1}\left\{ #2 \right\}}
\newcommand{\CPP}[3][]
{\ifthenelse{\equal{#1}{}}{{\mathbb P}\left(\left. #2 \, \right| #3 \right)}{{\mathbb P}_{#1}\left(\left. #2 \, \right | #3 \right)}}

\newcommandx{\osc}[2][1=]{\mathrm{osc}_{#1}(#2)}

\def\Idd{\operatorname{I}_d}

\newcommand\coupling[2]{\Gamma(\mu,\nu)}

\newcommand\restr[2]{{
  \left.\kern-\nulldelimiterspace 
  #1 
  \vphantom{\big|} 
  \right|_{#2} 
  }}

\newcommand\unifDist{\mathbf{Unif}}

\def\shift{\theta}

\def\proj{\operatorname{proj}}

\def\gaussLaw{\mathrm{N}}

\newcommand{\mb}[1]{\mathbf{#1}}

\def\Qmala{R^{\scriptscriptstyle{\mathsf{MALA}}}}
\def\Qrwm{R^{\scriptscriptstyle{\mathsf{RWM}}}}

\def\HMC{\scriptscriptstyle{\mathsf{HMC}}}

\def\MH{{\scriptscriptstyle{\mathsf{MH}}}}
\def\propMH{K^{\scriptscriptstyle{\mathsf{MH}}}}
\def\QMH{Q^{\MH}}
\def\alphaMH{\alpha^{\MH}}

\def\alphaMHD{\tilde{\alpha}^{\scriptscriptstyle{\mathsf{MH}}}}
\def\Leb{\lambda_{\scriptscriptstyle{\mathsf{Leb}}}}

\endlocaldefs

\numberwithin{equation}{section}

\begin{document}

\begin{frontmatter}
\title{Boost your favorite Markov Chain Monte Carlo sampler using Kac's theorem: the Kick-Kac teleportation algorithm}
\runtitle{Boost your favorite MCMC sampler using Kac's theorem: the KKT algorithm}

\begin{aug}
\author[A]{\fnms{Randal} \snm{Douc}\ead[label=e1]{randal.douc@telecom-sudparis.eu}},
\author[B]{\fnms{Alain} \snm{Oliviero Durmus}\ead[label=e2]{alain.durmus@polytechnique.edu}}
\\
\author[C]{\fnms{Aur\'elien} \snm{Enfroy}\ead[label=e3]{aurelien.enfroy@ens-paris-saclay.fr}}
\and
\author[D]{\fnms{Jimmy} \snm{Olsson}\ead[label=e4]{jimmyol@kth.se}}
\address[A]{Department CITI, Telecom SudParis, Evry, France, \printead{e1}}
\address[B]{CMAP, CNRS, Ecole Polytechnique, Institut Polytechnique de Paris, 91120 Palaiseau, France, \printead{e2}}
\address[C]{Laboratoire de mathématiques d’Orsay, Université Paris-Saclay, \\ Orsay, France, \printead{e3}}
\address[D]{Department of Mathematics, KTH Royal Institute of Technology, Stockholm, Sweden, \printead{e4}}
\end{aug}

\begin{abstract}
The present paper focuses on the problem of sampling from a given target distribution $\pi$ defined on some general state space. To this end, we introduce a novel class of non-reversible Markov chains, each chain being defined on an extended state space and having an invariant probability measure admitting $\pi$ as a marginal distribution. The proposed methodology is inspired by a new formulation of Kac's theorem and allows global and local dynamics to be smoothly combined. Under mild conditions, the corresponding Markov transition kernel can be shown to be irreducible and Harris recurrent. In addition, we establish that geometric ergodicity holds under appropriate conditions on the global and local dynamics. Finally, we illustrate numerically the use of the proposed method and its potential benefits in comparison to existing Markov chain Monte Carlo (MCMC) algorithms. 
\end{abstract}


\end{frontmatter}


\section{Introduction}
\label{sec:intro}
In this work, we develop a new Monte Carlo technique to sample from a
given target distribution $\pi$ on some general state space
$(\Eset, \Esigma)$. Our construction is based on \emph{Kac's theorem},
which is a fundamental result in Markov chain theory. More precisely,
let $P$ be some Markov kernel leaving $\pi$ invariant and
$\Cset \in \Esigma$ some measurable set. Moreover, for every $x \in \Esigma$, let $\PP_x$ be the law of the canonical Markov chain $\seq{X_n}{n \in \nset}$ on $(\Eset, \Esigma)$ with kernel $P$ and such that $X_0 = x$, and denote by $\PE_x$ the corresponding expectation. Under some conditions on $P$ and $\Cset$, Kac's theorem states that for every bounded Borel-measurable function $h$ on $(\Eset, \Esigma)$, the Lebesgue integral of $h$ with respect to $\pi$, denoted by $\pi(h)$, can be expressed as
\begin{equation} \label{eq:Kacs:formula:intro:display}
	\pi(h) = \pi(\Cset) \int \pi_\Cset(\rmd x) \, \PE_x \lrb{\sum_{k = 0}^{\rti_{\msc} - 1} h(X_k)}, 
\end{equation}
where $\sigma_{\msc} = \inf \{ k > 0 : X_k \in \msc\}$ is the first return time to $\Cset$ and $\pi_\Cset \eqdef \pi(\cdot \cap \Cset) / \pi(\Cset)$ is the (normalized) restriction of $\pi$ to $\Cset$ (see \Cref{thm:kac} below). Here we have assumed that $ \pi(\Cset) > 0$. A remarkable feature of Kac's formula is that it characterizes completely $\pi$ by $\pi(\msc)$, $\pi_{\msc}$ and $P$. In particular, there is a large flexibility in the choice of these three elements. In the light of \eqref{eq:Kacs:formula:intro:display}, a possible strategy for estimating $\pi(h)$ is to produce i.i.d. draws $(\X{0}{i})_{i=1}^N$ from $\pi_\Cset$ and let each draw $\X{0}{i}$ initialize a Markov chain $(\X{k}{i})_{k \in \nset}$ evolving according to $P$; once the chain returns to $\Cset$, {\ie}, after $\rtiC{i}$ time steps, it is killed and the last state is discarded. By \eqref{eq:Kacs:formula:intro:display}, the independent Markovian trajectories $(\X{k}{i})_{k = 0}^{\sigma_{\msc}^i - 1}$, 
$i \in \{1, \ldots, N\}$, 
produced in this manner can be used to form a consistent estimator $\sum_{i=1}^N \sum_{k=0}^{\rtiC{i}-1} h(\X{k}{i}) / \sum_{i = 1}^N \rtiC{i}$ of $\pi(h)$. More generally, instead of sampling i.i.d draws from $\pi_\Cset$, we may generate $(\X{0}{i})_{i \in \nset}$ by simulating a Markov chain evolving on $\Cset$ according to some Markov kernel $Q$ leaving $\pi_\Cset$ invariant. Then, as before, we let each state $\X{0}{i}$ initialize a Markovian excursion $(\X{k}{i})_{k = 0}^{\rtiC{i} - 1}$ evolving according to $P$. After this, the---now dependent---trajectories $(\X{k}{i})_{k = 0}^{\rtiC{i} - 1}$, $i \in \nset$, can be used to form an estimator of $\pi(h)$ as previously. Now, concatenate the produced trajectories into a single sequence and denote by $Y_n$ the $(n + 1)$th element of this sequence (or, formally, for every $n \in \nset$, let $i(n) \eqdef \sup \{i \in \nset : \sum_{\ell = 1}^i \rtiC{\ell} \leq n \}$, $k(n) \eqdef n - \sum_{\ell = 1}^{i(n)} \rtiC{\ell}$ and set $Y_n \eqdef \X{k(n)}{i(n)}$). The resulting process $(Y_n)_{n \in \nset}$, which we will refer to as the \emph{Kick-Kac teleportation} (KKT) process and define carefully in \Cref{sec:telep-proc}, is non-Markovian and evolves by moving---or, \emph{teleporting}---across $\Cset$ according to $Q$ and excursing the space outside $\Cset$ according to $P$. After each $Q$ transition, an excursion according to $P$ is initiated; once the excursion returns to $\Cset$, its last state is discarded, and a new teleportative move according to $Q$, picking up the end state of the previous teleportation, takes place, generating the new state of the process. When the density of $\pi$ is known up to a normalizing constant, so is the density of $\pi_\Cset$, which provides a variety of possible choices of $Q$, including the plethora of \emph{Markov chain Monte Carlo} (MCMC) kernels, such as those of \emph{Metropolis--Hastings} (MH) type \cite{metropolis:rosenbluth:rosenbluth:teller:teller:1953,hastings:1970,tierney:1994}. Moreover, with an appropriate choice of $\Cset$, \emph{rejection sampling} can be implemented to generate i.i.d. draws from $\pi_\Cset$; in that case, $Q(x, \cdot) = \pi_\Cset$ for all $x \in \Cset$, yielding a process lacking long-term memory. 

    The KKT sampling method that we propose has particular relevance to
    target distributions with multiple, isolated modes. Standard MCMC methods,  such as the \emph{random-walk Metropolis} (RWM) \emph{algorithm}, the \emph{Metropolis-adjusted Langevin algorithm} (MALA) \cite{roberts:tweedie:1996} or the 
\emph{Hamiltonian Monte Carlo} (HMC) \emph{method} \cite{neal:2011}, generally exhibit poor performance when
    applied to such models due to difficulties in crossing
    low-probability barriers. Still, since such models arise in many
    applications
    \cite{kou2006equi,feroz:hobson:bridges:2009,ihler2005nonparametric}, plenty of variations on the MCMC technology have been
    proposed to address this issue. The literature on this topic is
    extensive, why we here only provide a brief introduction to place our
    work into a context. A first family of methods is based on tempering
    strategies \cite{geyer1991markov,miasojedow2013adaptive,neal1996sampling} and a
    second one on optimization methods trying first to identify
    the modes of the target distribution before applying adaptive
    MCMC techniques
    \cite{andricioaei2001smart,lan2014wormhole}. Finally, exploratory
    schemes have been suggested in
    \cite{wang2001determining,bornn2013adaptive}. In contrast to these works, our methodology belongs to a class of
    `resurrected' or `regenerative' processes, which return to a
    certain part of the state space once being killed
    \cite{hobert2002applicability,hobert:robert:2004,mykland1995regeneration}. 
    In many cases, these processes have been constructed with the aim of coupling from the past and perfect simulation \cite{propp1996exact,fill1997interruptible}.
    Among existing methods, we found that the hybrid kernel approach introduced in 
    \cite{brockwell:kadane:2005} is the one that is most closely related to the KKT sampler proposed by us, and a detailed comparison is given in
    \Cref{sec:revis-hybr-kern} below. Using regeneration strategies for
    designing efficient Monte Carlo samplers is still the subject of
    active research; see
    \cite{jacob2020unbiased,lee2014perfect,wang2021regeneration} for
    some recent work in this direction. In
the case where $\pi$ is multi-modal, the KKT process provides, especially in its most
    general form---the \emph{general KKT process}---, a very
    generic framework. Using the teleportation-based sampling
    procedure described above, multi-modality can be straightforwardly
    handled by letting $\Cset$ be a low-probability region surrounding
    the modes of the target and $P$ some MCMC kernel used to explore
    locally the modes. For instance, $\Cset$ can be defined as a level
    set of the target density (up to a constant of
    proportionality). Parameterized appropriately, the KKT process can
    be made to transit with ease across regions of the state space
    where the ergodic properties of the kernel $P$ are poor,
    de-correlating significantly the produced samples. However, although we provide, in \Cref{sec:numerics}, a couple of numerical examples that well illustrate the potential of the proposed sampling methodology, a more far-reaching exploration of different ways of designing the set $\Cset$ and the kernel $Q$ is beyond the scope of the present paper and left as future research. 
      
In the present paper, which provides a basis for further methodological innovation, we 
\begin{itemize}
    \item provide, under mild assumptions, a novel and concise proof of Kac's theorem \eqref{eq:Kacs:formula:intro:display} as well as a generalized version of the same identity where $\pi_\Cset$ is replaced by a general probability distribution $\tilde{\pi}$ on $(\Eset, \Esigma)$ such that $\tilde{\pi}$ is absolutely continuous with respect to $\pi$ with a $\pi$-a.s. bounded density $\rmd \tilde{\pi} / \rmd \pi$. These results are of independent interest. 
    \item use Kac's theorem to construct an increasingly general KKT
      process framework and provide a rigorous analysis of its
      theoretical properties.
    \item show that our methodology extends the classcial MCMC
      framework in the sense that the general KKT process covers the
      MH algorithm as a special case.
    \item illustrate, in a simulation study, the potential of the proposed KKT sampler.   
\end{itemize}
  
\subsection*{Outline of the paper and the results} 

As a preparation for coming developments, we first provide, in \Cref{sec:kac}, our new proof of Kac's theorem \eqref{eq:Kacs:formula:intro:display} under the mild assumption that $\Cset$ is $\pi$-accessible (\Cref{thm:kac}). Moreover, as established by \Cref{cor:kac}, in the case where $\pi$ is the unique invariant distribution of $P$, Kac's theorem holds true for any set $\Cset \in \Esigma$ such that $\pi(\Cset) > 0$. 

In the first part of \Cref{sec:telep-proc} we examine the KKT process in the memoryless case where $Q(x, \cdot) = \pi_\Cset$ for all $x \in \Cset$. In this case, the KKT process $\seq{Y_n}{n \in \nset}$ is a Markov chain, and we establish that it allows the target $\pi$ of interest as an invariant distribution. In addition, we provide necessary and sufficient conditions for the process to be $\pi$-reversible (\Cref{prop:nonrev}), conditions that are not satisfied in general. 

In the second part of \Cref{sec:telep-proc} we turn to a general $Q$ leaving $\pi_{\msc}$ invariant, and define the KKT process $\seq{Y_n}{n \in \nset}$ via an additional, auxiliary $\Cset$-valued process $(Z_n)_{n \in \nset}$ defining the states of $\seq{Y_n}{n \in \nset}$ inside $\Cset$. In this construction (which is detailed in \Cref{alg:telep}), the bivariate process $(Y_n, Z_n)_{n \in \nset}$ is a Markov chain on the extended state space $\Eset \times \Cset$ with invariant probability measure $\check{\pi}$ admitting $\pi$ as a marginal with respect to the first argument (\Cref{lem:invariant_measure_kac_process_c}). Moreover, if $\pi$ and $\pi_\Cset$ are the unique invariant probability measures for $P$ and $Q$, respectively, then $\check{\pi}$ is the unique invariant probability measure for the process $(Y_n, Z_n)_{n \in \nset}$ (\Cref{propo:uniqueness_bar_pi}). Then, under the assumption that $\check{\pi}$ is the unique invariant probability measure, we show that $(Y_n, Z_n)_{n \in \nset}$ satisfies a law of large numbers with respect to $\check{\pi}$ if and only if the durations of the excursions from $\Cset$ are a.s. finite and $Q$ satisfies a law of large numbers with respect to $\pi_\Cset$ (\Cref{thm:LLN_Markov}). 

\sloppy In order to complete the analysis of the KKT process, we also show
that $(Y_n, Z_n)_{n \in \nset}$ is geometrically ergodic (with respect
to $\check{\pi}$) if $Q$ is geometrically ergodic (with respect to
$\pi_\Cset$) and $P$ exhibits geometric drift toward $\Cset$
(\Cref{thm:ergo:geom}). We show that the latter automatically holds in
the special case where the complement of $\msc$ is compact and $P$ is
either an MH or a Feller kernel.  Importantly, $P$ is not required to be geometrically
ergodic, which supports the idea underpinning the construction, namely
that the mixing of a Markov chain evolving according to $P$ can be
drastically improved by changing the dynamics of the chain in regions
of bad mixing.

In \Cref{sec:extensions}, we present the generalization of Kac's theorem described above (\Cref{prop:kac:general}) and use this result to construct a general KKT process which can be shown to cover the classical MH algorithm as well as the hybrid kernel of \cite{brockwell:kadane:2005} as special cases. 

The proposed sampling technology is illustrated numerically in \Cref{sec:numerics}, where the KKT sampler is benchmarked successfully against some existing advanced MCMC algorithms. 

Finally, to facilitate reading, the notation used in this paper is collected in \Cref{sec:notation} and some technical results and proofs are postponed to \Cref{sec:appendix}.


\section{Kac's formula}
\label{sec:kac}
As explained above, the teleportation process introduced in \Cref{sec:telep-proc} is based on Kac's theorem. Due to its importance to the construction of the teleportation process, we restate, after the introduction of adequate notation, this result and provide general assumptions under which the theorem holds true. In addition, we provide an effective and, up to our knowledge, novel proof. 

Let $(\Eset, \Esigma)$ be some measurable space and denote by $\measureset(\Esigma)$ and $\measureset_1(\Esigma)$ the sets of $\sigma$-finite nonnegative measures and probability measures, respectively, on $(\mse,\mce)$. In addition, we denote by $\bddfunc{\Esigma}$ and $\nonnegfunc{\Esigma}$ the sets of bounded and nonnegative $\mce/\mcbb(\rset)$-measurable functions, respectively. Furthermore, let $(\Eset^\nset, \Esigma^{\tensprod \nset})$ be the associated canonical space and denote by $\seq{X_n}{n \in \nset}$ the corresponding coordinate process. For every initial distribution $\Xinit \in \measureset_1(\Esigma)$ and Markov kernel $P$ on $\Eset \times \Esigma$ we let $\PP_\Xinit$ be the probability measure on $(\Eset^\nset, \Esigma^{\tensprod \nset})$ induced by $P$ and $\Xinit$ and denote by $\PE_\Xinit$ the corresponding expectation. As usual, the \emph{shift operator} is the transformation on $\Eset^\nset$ given by 
$\shift : (\omega_0, \omega_1, \omega_2, \ldots) \mapsto (\omega_1, \omega_2, \ldots)$. If $\pi$ is an invariant probability measure for $P$, \ie~$\pi P=\pi$, then by the Markov property, for every $\msa \in \Esigma^{\tensprod \nset}$,
\begin{equation} \label{eq:measure_preserving}
    \PE_\pi[\indi{\msa} \circ \shift] = \PE_\pi[\indi{\msa}] \eqsp, 
\end{equation}
which implies that $\PP_\pi \circ \shift^{-1} = \PP_\pi$; in other words, $(\Eset^\nset, \Esigma^{\tensprod \nset}, \PP_\pi, \shift)$ is a measure-preserving dynamical system. Moreover, for every $\Cset \in \Esigma$ we define recursively 
\begin{equation} \label{eq:def_return_time_X}
    \rti^\ell_\msc \eqdef \inf\set{k > \rti_\Cset^{\ell-1}}{X_k\in \Cset}\eqsp, \quad \ell \in \nsetpos\eqsp, 
\end{equation}
with $\rti_\Cset^0 \eqdef 0$. Note that $\rti_\Cset \eqdef \rti_\Cset^1$ is the return time to the set $\msc$. In the following we say that $\Cset\in \Esigma$ is \emph{$\pi$-accessible} if $\PP_x(\rti_\Cset < \plusinfty) > 0$ for $\pi$-almost all $x \in \Eset$. Moreover, we simply say that $\Cset$ is \emph{accessible} if $\PP_x(\rti_\Cset < \plusinfty) > 0$ for all $x \in \Eset$.

We have now all the notation required for the statement of Kac's theorem.
\begin{theorem}[Kac's theorem] \label{thm:kac} 
Let $P$ be a Markov kernel on $\Eset \times \Esigma$
with invariant probability measure $\pi$. Then for every $\pi$-accessible set $\msc \in \Esigma$ it holds that 
\begin{equation} \label{eq:kac:formula}
    \pi = \pi^0_\msc = \pi^1_\msc,
\end{equation}
where the measures $\pi^0_\msc$ and $\pi^1_\msc$ are defined by, for $f \in \nonnegfunc{\Esigma}$, 
\begin{equation} \label{def:pi0}
\pi^0_\msc(f) \eqdef \int_\msc \pi(\rmd x) \,  \PE_x \lrb{\sum_{k=0}^{\rti_\msc-1}f(X_k)}, \quad \pi^1_\msc(f) \eqdef \int_\msc \pi(\rmd x) \, \PE_x\lrb{\sum_{k=1}^{\rti_\msc} f(X_k)}.
\end{equation}
\end{theorem}

Even though the conclusion of \Cref{thm:kac} is the same as that of \cite[Theorem~10.4.9]{meyn:tweedie:2012}, these results are obtained under different assumptions. Indeed, \cite[Theorem 10.4.9]{meyn:tweedie:2012} requires the Markov kernel $P$ to be recurrent and irreducible, properties that are not assumed in \Cref{thm:kac}. It can also be noted that \cite[Theorem~3.6.5]{douc:moulines:priouret:soulier:2018} provides another statement of Kac's theorem which is expressed in terms of subinvariant measures (instead of invariant probability measures). Consequently, in \cite{douc:moulines:priouret:soulier:2018}, the identity \eqref{eq:kac:formula} holds under the assumptions that (i) $\Cset$ is $\pi$-accessible and (ii) $\PP_x(\sigma_\Cset<\infty)=1$ for $\pi$-almost all $x \in \Cset$, whereas, as we will see in \Cref{lem:probOne} below, it turns out that that (ii) is already implied by (i) in the particular case of invariant probability measures. This allows the proof of \eqref{eq:kac:formula} to be neatly simplified.   

We provide here an intermediate result which will be used in the proof of \Cref{thm:kac}

\begin{lemma} \label{lem:probOne}
    Let $P$ be a Markov kernel on $\Eset \times \Esigma$ with invariant probability measure $\pi$. Then $\Cset \in \Esigma$ is $\pi$-accessible if and only if $\PP_x(\rti_\msc < \plusinfty) = 1$ for $\pi$-almost all $x \in \Eset$.
\end{lemma}

\begin{proof}
    Assume that $\msc$ is $\pi$-accessible and set $\msa \eqdef \{\rti_\msc < \plusinfty\}$. Since $\{\rti_\msc \circ \shift<\plusinfty\} \subset \{\rti_\msc<\plusinfty\}$, it holds that $\indi{\msa} \circ \shift\leq \indi{\msa}$. Combining this with \eqref{eq:measure_preserving} yields $\indi{\msa} \circ \shift = \indi{\msa}$, $\PP_\pi$-a.s., \ie, $\indi{\msa}$ is $\PP_\pi$-a.s. an invariant random variable for the dynamical system $(\Eset^\nset, \Esigma^{\varotimes \nset}, \PP_\pi, \shift)$. By \cite[Theorem~17.1.1]{meyn:tweedie:2012}, $\PP_{X_0}(\msa) = \indi{\msa}$, $\PP_\pi$-a.s. Since $\msc$ is $\pi$-accessible by assumption, it holds that $\PP_x(\msa)>0$ for $\pi$-almost all $x \in \Eset$; thus, we obtain 
    \begin{equation*}
        1 = \int \pi(\rmd x) \, \indiacc{\PP_x(\msa) > 0} = \PP_\pi(\PP_{X_0}(\msa) > 0)
        = \PP_\pi(\indi{\msa} > 0) = \PP_\pi(\msa) = \PP_\pi(\rti_\msc < \plusinfty).
    \end{equation*}
    Consequently, 
    $$
        0 = 1 - \PP_\pi(\rti_\msc<\plusinfty) = \int \pi(\rmd x) \lr{1 - \PP_x(\rti_\Cset<\plusinfty)},
    $$
    from which we finally conclude that $\PP_x(\rti_\Cset<\plusinfty)=1$ for $\pi$-almost all $x \in \Eset$. The converse is obvious.
\end{proof}

\begin{proof}[Proof of \Cref{thm:kac}]
By the last-exit decomposition and the Markov property we have for all $f \in \posfunc{\Esigma} \cap \bddfunc{\Esigma}$ and $n \in \nsetpos$, 
\begin{align*}
    \pi(f) &
    = \sum_{\ell=1}^n \PE_\pi \lrb{f(X_n) \indi{\msc}(X_\ell) \prod_{k=\ell+1}^n \indi{\msc^c}(X_{k})}+ \PE_\pi[f(X_n) \indin{\rti_\msc>n}] \\
    &= \sum_{\ell=1}^n \PE_\pi \lrb{\indi{\msc}(X_\ell) \PE_{X_{\ell}}\lrb{f(X_{n-\ell}) \prod_{k=1}^{n-\ell}\indi{\msc^c}(X_{k})}} + \PE_\pi[f(X_n) \indin{\rti_\msc>n}]\eqsp.
\end{align*}
Noting that $\pi$ is invariant and setting $k=n-\ell$, we get
\begin{align}
    \pi(f) &= \sum_{k=0}^{n-1} \int_\msc \pi(\rmd x) \, \PE_x[f(X_k) \indin{\rti_\msc>k}]+ \PE_\pi[f(X_n) \indin{\rti_\msc>n}] \nonumber \\
    &= \int_\msc \pi(\rmd x) \, \PE_x \lrb{\sum_{k = 0}^{(n - 1) \wedge (\rti_\msc -1)} f(X_k)} + \PE_\pi[f(X_n) \indin{\rti_\msc>n}] \eqsp. \label{eq:kac:one}
\end{align}
We now let $n \to \infty$ in \eqref{eq:kac:one}. Since $f$ is bounded and $\PP_\pi(\rti_\msc=\plusinfty)=0$ by~\Cref{lem:probOne}, the dominated convergence theorem implies that the second term on the right-hand side of \eqref{eq:kac:one} tends to zero. Moreover, by applying the monotone convergence theorem to the first term, we obtain that $\pi = \pi^0_\msc$. Finally, this implies the last equality of \eqref{eq:kac:formula}, since by the Markov property, $\pi^1_\msc=\pi^0_\msc P=\pi P= \pi$.
\end{proof}

\begin{remark} \label{rem:arbitrary}
  If $C \in \Esigma$ is arbitrary (and not necessarily $\pi$-accessible), then by \eqref{eq:kac:one}, $\pi(f) \geq \int_\msc \pi(\rmd x) \, \PE_x [\sum_{k = 0}^{(n - 1) \wedge (\rti_\msc -1)} f(X_k)]$ for every $f \in \posfunc{\Esigma}$. Thus, the monotone convergence theorem implies that
  \begin{equation*}
  \pi(f) \geq \lim_{n\to \infty} \int_\msc \pi(\rmd x) \, \PE_x \lrb{\sum_{k = 0}^{(n - 1) \wedge (\rti_\msc -1)} f(X_k)} = \int_\msc \pi(\rmd x) \, \PE_x \lrb{\sum_{k = 0}^{\rti_\msc - 1}f(X_k)}. 
  \end{equation*} 
\end{remark}

\begin{remark} \label{rem:positive}
If $\msc \in \Esigma$ is $\pi$-accessible, then \Cref{thm:kac} implies that $\pi(\msc)>0$; indeed, otherwise $\pi_\msc^0$, and hence $\pi$, would be the null measure, which is not possible. 
\end{remark}

To apply \Cref{thm:kac}, we need to check that  $\msc$ is $\pi$-accessible. If the Markov kernel $P$ has a unique invariant probability measure $\pi$, it turns out that this condition can be simplified.

\begin{lemma} \label{lem:verif_C_acc_unique_invariant}
    Let $P$ be a Markov kernel on $\Eset \times \Esigma$ with a unique invariant probability measure $\pi$. Then every set $\Cset \in \Esigma$ such that $\pi(\Cset) > 0$ is $\pi$-accessible.
\end{lemma}
Before proving \Cref{lem:verif_C_acc_unique_invariant}, we combine it with \Cref{thm:kac}
to obtain the following corollary, which allows Kac's
theorem to be applied to every set $\Cset \in \Esigma$ such that $\pi(\Cset) > 0$. 
\begin{corollary} \label{cor:kac}
    Let $P$ be a Markov kernel on $\Eset \times \Esigma$
    with a unique invariant probability measure $\pi$.
    Then for all $\msc \in \Esigma$ such that $\pi(\msc) > 0$, $\pi = \pi^0_\msc = \pi^1_\msc$,
    where the measures $\pi^0_\msc$ and $\pi^1_\msc$ are defined in \eqref{def:pi0}.
\end{corollary}

\begin{proof}[Proof of \Cref{lem:verif_C_acc_unique_invariant}] 
    Let $\msc \in \Esigma$ be such that $\pi(\msc) > 0$. Since $P$ has a unique invariant probability measure, \cite[Theorem~5.2.6]{douc:moulines:priouret:soulier:2018} implies that the associated dynamical system is ergodic. Therefore, by the Birkhoff ergodic theorem,  
    $$
    \lim_{n \to \plusinfty} n^{-1} \sum_{k = 0}^{n-1} \indi{\msc}(X_k) = \pi(\msc) > 0, \quad \mbox{$\PP_\pi$-a.s.}
    $$
    Thus, $\PP_\pi(\rti_\Cset <\plusinfty)=1$, implying that $\PP_x(\rti_\msc<\plusinfty) = 1 > 0$ for $\pi$-almost all $x \in \Eset$. This shows that $\msc$ is $\pi$-accessible.
\end{proof}


\section{The Kick-Kac teleportation process}
\label{sec:telep-proc}
Given  $\pi \in \measureset_1(\Esigma)$, we will now use Kac's theorem (\Cref{thm:kac}) to construct a stochastic process targeting $\pi$. In the following we let $P$ be some $\pi$-invariant Markov kernel on $\Eset \times \Esigma$; this kernel will be referred to as the \emph{base kernel}.

\subsection{The memoryless Kick-Kac teleportation process}

As we saw in the previous section, the following assumption guarantees that Kac's theorem (\Cref{thm:kac}) applies. 

\begin{hyp}{A}
    \item \label{ass:P:inv}
    The Markov kernel $P$ allows an invariant probability measure $\pi$ such that $\Cset$ is $\pi$-accessible.
\end{hyp}
As shown by \Cref{lem:verif_C_acc_unique_invariant},
\ref{ass:P:inv}
is satisfied if $\pi$ is the unique invariant probability measure of
$P$ and $\pi(\msc)>0$.

Under \ref{ass:P:inv}, the probability measure
\begin{equation}
  \label{eq:def_pi_C}
  \pi_\msc : \Esigma \ni \Aset \mapsto \pi(\Aset \cap \msc)/\pi(\msc) 
\end{equation}
is well defined in the light of \Cref{rem:positive} and provides the conditional probabilities of $\pi$ given $\msc$. Note that \eqref{eq:kac:formula}--\eqref{def:pi0} express the invariant probability measure $\pi$ in terms of its restriction to the set $\msc$ only. This remarkable fact is underpinning the construction of the \emph{memoryless Kick-Kac teleportation} (KKT) \emph{process}, whose evolution is described by \Cref{alg:memoryless:telep}. In words, the memoryless KKT process evolves by proposing a candidate according to the Markov kernel $P$. If the candidate falls
into the region $\msc$, it is discarded and the next state is drawn exactly from $\pi_\msc$, independently of the past (in this sense, the process is \emph{memoryless}); otherwise, if the candidate falls outside $\msc$, it is accepted as the new state of the process. 

\begin{algorithm}[h!]
    \caption{The memoryless Kick-Kac teleportation process}
    \label{alg:memoryless:telep}
    \begin{algorithmic}[1]
        \State {\bf Initialization:} draw $Y_0$
        \For{$k\gets 1$ to $n$}
            \State draw $Y_k^\star \sim P(Y_{k-1},\cdot)$
            \If {$Y_k^\star \notin \msc$}
                \State set $Y_{k} \gets Y_k^\star$
            \Else
                \State draw $Y_k\sim \pi_\msc$
            \EndIf
        \EndFor
    \end{algorithmic}
\end{algorithm}

Of course, the construction requires the set $\Cset$ to be chosen in such a way that exact sampling from $\pi_{\msc}$ is feasible. A natural approach is based on the accept-reject algorithm. More precisely, assume that the probability measure $\pi$ is absolutely continuous with respect to  some $\sigma$-finite measure $\nu$ and denote by $p_{\pi}$ a version of its Radon--Nikodym derivative with respect to $\nu$. In addition, let $q: \Eset \to \rset_+$ be another probability density function with respect to $\nu$, serving as instrumental density, and assume that $\msc \subset \set{x \in \Eset}{p_{\pi}(x) \leq \epsilon q(x)}$. Then, as $\pi_{\msc}$ admits a density proportional to $\1_\msc \, p_{\pi}$ with respect to $\nu$ and $\1_\msc (x) p_{\pi}(x) \leq \epsilon q(x)$ for every $x \in \Eset$, the accept-reject method can be applied to obtain \iid~samples from $\pi_{\msc}$, provided that the ratio $\1_{\msc}(x)p_{\pi}(x)/[\epsilon q(x)]$ is computable for all $x \in \Eset$ such that $q(x) >0$.

The Markov transition kernel $S$ on $\Eset \times \Esigma$ associated with the Markov chain $\seq{Y_n}{n \in \nset}$ is given by
$$
    S h(y) \eqdef \int_{\Cset^c} P(y,\rmd y') \, h(y') + P(y, \Cset) \pi_\Cset(h), \quad (y, h) \in \Eset \times \posfunc{\Esigma}.
$$
Since $\pi$ is invariant for $P$, it holds that  
\begin{equation*}
    \pi S h = \pi P(\indi{\Cset^c} h)+\pi P(\Cset) \pi_\Cset(h)=\pi(\indi{\Cset^c} h)+ \pi(\indi{\Cset}h)=\pi(h), 
\end{equation*}
which means that $\pi$ is also invariant for $S$ (some kernel notation, such as $S h$ and $\pi S$ above, that will be used in the rest of the paper is provided in \Cref{sec:notation}).

We now address the $\pi$-reversibility of $S$. In the following, we let $\pi \tensprod P$ be the probability measure in $\measureset_1(\Esigma^{\tensprod 2})$ defined by 
$ \pi \tensprod P(\Aset) \eqdef \int \1_{\Aset}(x, y) \, \pi(\rmd x) \, P(x, \rmd y)$, $\Aset \in \Esigma^{\tensprod  2}$. Note that with this notation, $P$ is reversible with respect
to $\pi$ if and only if for every $\Aset \in \Esigma$ and $\Bset \in \Esigma$, $ \pi \tensprod P(\Aset \times \Bset) =  \pi \tensprod P(\Bset \times \Aset)$. In addition, for any given $\msd \in \Esigma$, let $\Esigma_{\msd} \eqdef \set{\msd \cap \Aset}{\Aset \in\Esigma}$ be the trace $\sigma$-field of $\Esigma$ on $\msd$ and denote, for any measure $\nu \in \measureset_1(\Esigma)$, by $\nu |_{\msd}$ the restriction of $\nu$ to $\Esigma_\msd$. 

\begin{proposition} \label{prop:nonrev}
  Assume \ref{ass:P:inv} and that $\Esigma_{\Cset^c}$ is countably generated. Then $S$ is $\pi$-reversible if and only if the following two conditions are satisfied. 
  \begin{enumerate}[(a)]
    \item For every $\Aset \in \Esigma_{\Cset^c}^{\tensprod 2}$,   
    $$
    \int \hspace{-2mm} \int \indi{\Aset}(x,y) \, \pi(\rmd x) \, P(x,\rmd y) = \int \hspace{-2mm} \int  \indi{\Aset}(y,x) \, \pi(\rmd x) \, P(x, \rmd y). 
    $$ \label{item:nonrev:a} 
    \item \label{item:nonrev:b} There exists $\mu \in \measureset(\Esigma_{\Cset^c})$  
    such that $P(y,\cdot)|_{\Cset^c} = \mu$ for $\pi$-almost all $y\in\Cset$. 
  \end{enumerate} 
\end{proposition}

The proof of \Cref{prop:nonrev} is given in \Cref{sec:S:nonRev}. 

In the light of \Cref{prop:nonrev}, we may expect $S$ to be non-reversible in general; indeed, even if $P$ is $\pi$-reversible (implying \ref{item:nonrev:a}), condition \ref{item:nonrev:b} will typically not be satisfied. Some classes of MCMC methods based on non-reversible kernels have been shown to exhibit favorable convergence properties compared to standard reversible methods. For instance, \cite{diaconis:holmes:neal:2000} carries through a detailed study of a particular type of \textit{lifted Markov chains}, and shows theoretically and numerically that this class of MCMC algorithms provides better accuracy than standard reversible ones. In addition, by extending results in \cite{tierney:1998}, \cite{andrieu:livingstone:2019} establishes that non-reversible MCMC algorithms most often lead to improved asymptotic variance in the case where a central limit theorem holds. However, whereas the non-reversibility of the memoryless KKT algorithm is encouraging, the requirement of exact sampling from $\pi_\Cset$ may limit the choice of the set $\Cset$. In the next section we will show how this drawback can be circumvented by including an additional Markov transition kernel into the construction. 

\subsection{The Kick-Kac teleportation process}

To avoid exact sampling from the conditional probability measure $\pi_\msc$ defined by \eqref{eq:def_pi_C}, we now provide an alternative to the memoryless KKT process. The process that we will construct includes another level of transitions according to some Markov kernel $Q$ on $\msc \times \Esigma_\msc$ which will be referred to as the teleportation kernel and will be supposed to satisfy the following condition.
\begin{hyp}{A}
  \item \label{ass:Q:inv}
  The Markov transition kernel $Q$ on $\msc \times \Esigma_\msc$ allows $\pi_\Cset$ as invariant probability measure.
\end{hyp}

The output $\seq{Y_k}{k \in \nset}$ of \Cref{alg:telep}, in which the exact sampling from $\pi_\Cset$ in \Cref{alg:memoryless:telep} is replaced by Markovian moves according to $Q$, will be referred to as the \emph{Kick-Kack teleportation} (KKT) \emph{process} in the sequel. In \Cref{alg:telep}, a candidate is proposed according to the kernel $P$. If it falls outside $\Cset$, the candidate is accepted as the new state of the process. On the other hand, if the candidate falls into the critical region $\Cset$, it is discarded, and a new state of the process is generated through a transition according to $Q$, starting off from the last state in $\Cset$ in the past history of the process. In the case where $\Cset$ is a region where the ergodic behavior of $P$ is poor, the KKT process simply replaces transitions according to $P$ across $\Cset$ by transitions according to the kernel $Q$, which is taylor-made to target efficiently $\pi_{\msc}$.

\begin{algorithm}[h]
    \caption{The Kick-Kack teleportation (KKT) process}
    \label{alg:telep}
    \begin{algorithmic}[1]
        \State {\bf Initialization:} draw $(Y_0,Z_0)$
        \For{$k \gets 1$ to $n$}
            \State draw $Y^\star_k \sim P(Y_{k - 1}, \cdot)$
            \If {$Y^\star_k \notin \msc$}
                \State set $(Y_k, Z_k) \gets (Y^\star_k,Z_{k-1})$
            \Else
                \State draw $Z_k \sim Q(Z_{k - 1}, \cdot)$ and set $Y_k \gets Z_k$
            \EndIf
        \EndFor
    \end{algorithmic}
\end{algorithm}

While $\seq{Y_k}{k\in \nset}$ is not a Markov chain in general, so is $\seq{Y_k, Z_k}{k \in \nset}$. It is easily seen that the Markov transition kernel of the latter bivariate process is given by
\begin{equation}
  \label{eq:def_R_Kac}
  R h(y, z) \eqdef \int_{\Cset^c} P(y, \rmd y') \, h(y', z)+ P(y, \Cset) \int_{\Cset} Q(z, \rmd z') \, h(z',z'),  
\end{equation}
for $(y, z) \in \Eset \times \Cset$ and $h \in \posfunc{\Esigma \tensprod \Csigma}$. We now provide an expression of an invariant probability measure of $R$ which allows $\pi$ as a marginal. This constitutes a first step towards justifying theoretically that the KKT process can be used for producing approximate draws from $\pi$. 

From now on, we will deal simultaneously with several Markov kernels. Thus, we will specify the notation $\PP_\Xinit$ to avoid ambiguities. More precisely, for a given measurable space $(\Fset, \Fsigma)$, Markov kernel $K$ on $\Fset \times \Fsigma$ and initial distribution $\Xinit \in \measureset_1(\Fsigma)$ we denote by $\PP^K_\Xinit$ the probability measure on $(\Fset^\nset, \Fsigma^{\tensprod \nset})$ induced by $K$ and $\Xinit$. The associated expectation operator is denoted $\PE^K_\Xinit$. If $x\in \Fset$ and $\Xinit=\delta_x$, we write $\PPP{K}{x}$ and $\PEE{K}{x}$ instead of $\PPP{K}{\delta_x}$ and $\PEE{K}{\delta_x}$, respectively. 
Now, define the probability measure 
\begin{equation}
    \label{eq:df_tilde_pi_kac_c}
    \check{\pi}(h) \eqdef \int_{\msc} \pi(\rmd x)  \, \PEE{P}{x} \lrb{\sum_{k=0}^{\rti_{\msc}-1} h(X_k,x)}, \quad h \in \posfunc{\Esigma \tensprod \Csigma}.
\end{equation}
The following key result justifies theoretically the use of the KKT process for targeting $\pi$. 
\begin{proposition}
    \label{lem:invariant_measure_kac_process_c}
    Assume \ref{ass:P:inv}. Then, $\pi$ is the marginal of $\check{\pi}$ with respect to the first component. Moreover, if also \ref{ass:Q:inv} holds, then $\check{\pi}$ is an invariant probability measure for $R$.
\end{proposition}
\begin{proof}
Under \ref{ass:P:inv}, \Cref{thm:kac} applies to $\pi$. Thus, for every $\Aset \in\Esigma$, 
$$
    \check{\pi}(\Aset \times \Cset) = \int_{\msc} \pi(\rmd x) \, \PEE{P}{x}\lrb{\sum_{k = 0}^{\rti_{\msc} - 1} \indi{\Aset}(X_k)} = \pi(\Aset)\eqsp, 
$$
which establishes the first claim of the proposition. 

We turn to the second claim. Under \ref{ass:P:inv}, \Cref{lem:probOne} shows that $\PPP{P}{x}(\rti_\Cset<\plusinfty)=1$ for $\pi$-almost all $x\in \Eset$. This fundamental property will be used in several parts of the proof. By the very definitions of $\check{\pi}$ and $R$, for every $h \in \posfunc{\Esigma \tensprod \Csigma}$,
\begin{equation}\label{eq:barpiR}
\check{\pi} R h=\int_{\Cset} \pi(\rmd x) \, \PEE{P}{x}\lrb{\sum_{k=0}^{\rti_\Cset-1} Rh(X_k,x)}=A+B,
\end{equation}
where
\begin{align*}
A&\eqdef \int_\Cset \pi(\rmd x) \, \PEE{P}{x} \lrb{\sum_{k = 0}^{\rti_\Cset -1} \int_{\Cset^c}P(X_k,\rmd y) \, h(y, x)}, \\
B&\eqdef \int_\Cset \pi(\rmd x) \, \PEE{P}{x}\lrb{\sum_{k = 0}^{\rti_\Cset -1} P(X_k,\Cset)} \int_{\msc} Q(x,\rmd z) \, h(z, z).
\end{align*}
Using the Markov property (in combination with Tonelli's theorem), we may write 
\begin{align*}
  A&= \int_\Cset \pi(\rmd x) \, \PEE{P}{x} \lrb{\PEE{P}{x}\parentheseDeux{ h(X_{1},x)\indi{\Cset^c}(X_{1})}+\sum_{k=1}^{\infty}  \indin{\rti_{\Cset} > k}\PEE{P}{X_k}\parentheseDeux{ h(X_{1},x)\indi{\Cset^c}(X_{1})}}\\
   &=\int_\Cset \pi(\rmd x) \, \PEE{P}{x} \lrb{ h(X_{1},x)\indi{\Cset^c}(X_{1})+\sum_{k=1}^{\infty}  \PEE{P}{x}\parentheseDeux{ h(X_{k+1},x)\indin{\rti_{\Cset} > k}\indi{\Cset^c}(X_{k+1}) \mid \mcg_k } }\\
     &=\int_\Cset \pi(\rmd x) \, \PEE{P}{x} \lrb{ h(X_{1},x)\indi{\Cset^c}(X_{1})+\sum_{k=2}^{\rti_{\msc}}  h(X_{k},x)\indi{\Cset^c}(X_{k}) }\\
    &=\int_\Cset \pi(\rmd x) \, \PEE{P}{x}\lrb{\sum_{k=1}^{\rti_\Cset}  \indi{\Cset^c}(X_k) h(X_k, x)}=\int_\Cset \pi(\rmd x) \, \PEE{P}{x}\lrb{\sum_{k=1}^{\rti_\Cset-1}  h(X_k, x)},
\end{align*}
where for any $k\in \nsets$, $\mcg_k=\sigma(X_1,\ldots,X_k)$ and the last equality follows as $\indi{\Cset^c}(X_{k})=1$ for all $k \in \{1,\ldots,\rti_\Cset-1\}$ and $\indi{\Cset^c}(X_{\rti_\Cset}) = 0$ on the set $\{\rti_\Cset < \plusinfty\}$ (having $\PPP{P}{x}$-probability one for $\pi$-almost all $x \in \Eset$). Similarly, using again the Markov property, 
\begin{align*}
  B&=\int_\Cset \pi(\rmd x) \, \PEE{P}{x}\lrb{\sum_{k = 0}^{\rti_\Cset -1}  \indi{\Cset}(X_{k+1})} \int_\Cset Q(x,\rmd z) \, h(z, z) \\
  &=\int_\Cset \pi(\rmd x) \, \PEE{P}{x}\lrb{\sum_{k = 1}^{\rti_\Cset} \indi{\Cset}(X_k)} \int_\Cset Q(x,\rmd z) \, h(z, z) \\ &=\int_\Cset \pi(\rmd x) \int_\Cset Q(x, \rmd z) \, h(z, z) \eqsp,
\end{align*}
where the last equality follows since $\indi{\Cset}(X_{k}) = 0$ for all $k \in \{1,\ldots,\rti_\Cset-1\}$ and $\indi{\Cset}(X_{\rti_\Cset})=1$ on $\{\rti_\Cset <\plusinfty\}$. Therefore,   \ref{ass:Q:inv} implies that $B  =\int_\Cset \pi(\rmd z) \, h(z, z)$. Plugging the obtained expressions of $A$ and $B$ into \eqref{eq:barpiR} yields
$$
    \check{\pi} R h = \int_\Cset \pi(\rmd x) \, \PEE{P}{x}\lrb{\sum_{k=0}^{\rti_\Cset-1}  h(X_k, x)} = \check{\pi}(h)\eqsp, 
$$
showing that $\check{\pi}$ is invariant for $R$. 
\end{proof}

\subsection{A law of large numbers for the KKT process}

In order to establish that  $\seq{Y_k, Z_k}{k \in \nset}$, evolving according to the Markov kernel $R$, satisfies a law of large numbers, we will first find conditions under which $R$ admits a unique invariant probability measure. In that case, \cite[Theorem~5.2.6]{douc:moulines:priouret:soulier:2018} guarantees that the associated dynamical system is ergodic and that Birkhoff's theorem applies.  
\begin{proposition}
    \label{propo:uniqueness_bar_pi}
    Assume \ref{ass:P:inv} and \ref{ass:Q:inv}. In addition, suppose that $\pi_{\msc}$ and $\pi$ are the unique invariant probability measures for $Q$ and $P$, respectively. Then $\check{\pi}$ is the unique invariant probability measure for $R$. 
\end{proposition}

\begin{proof}
  By \Cref{lem:invariant_measure_kac_process_c}, $\check{\pi}$ is an invariant probability measure for $R$. We show by contradiction that it is the unique invariant probability measure for $R$. First, assume that if $\mu$ is an other invariant probability measure for $R$, then there exists $c >0$ such that for every $\msa \in \mce$,  
  \begin{equation}
    \label{eq:proof_domination_invariant_R}
  \check{\pi}(\msa) \leq c \mu(\msa).  
  \end{equation}
   Now, if $R$ does not admit a unique probability measure, then there exist, by \cite[Theorem~1.4.6(ii)]{douc:moulines:priouret:soulier:2018}, at least two distinct probability measures $\mu_1$ and $\mu_2$ that are mutually singular, \ie, there exists $\Aset \in \Esigma \tensprod \Csigma$ such that $\mu_1(\Aset)=\mu_2(\Aset^c)=0$. Since $\mu_1$ and $\mu_2$ satisfy \eqref{eq:proof_domination_invariant_R} for some constants $c_1$ and $c_2$, we obtain that $\check{\pi}(\Aset)=\check{\pi}(\Aset^c)=0$, which, finally, implies that $\check{\pi}(\Eset \times \Cset) = \check{\pi}(\Aset)+\check{\pi}(\Aset^c)=0$. This is contradictory since $\check{\pi}$ is a probability measure.

   Now, we let $\mu \in \measureset_1(\Esigma \tensprod \Csigma)$ be any invariant probability measure for $R$ and show \eqref{eq:proof_domination_invariant_R} holds for some $c >0$.
Since for every $\Bset \in \Csigma$, 
$$
    \mu(\Eset \times \Bset) = \mu R (\Eset \times \Bset) = \int_{\Eset \times \Bset} \mu(\rmd (y, z)) \, P(y, \msc^c) + \int_{\Eset \times \msc} \mu(\rmd (y, z)) \, P(y, \msc) Q(z, \Bset)\eqsp, 
$$
it holds that 
\begin{equation*}
    \int_{\Eset \times \Bset} \mu(\rmd (y, z)) \, P(y, \msc) = \int_{\mse\times \msc} \mu(\rmd (y, z)) P(y, \msc) Q(z,\Bset)\eqsp. 
\end{equation*}
In other words, the finite nonnegative measure $\Csigma \ni \Bset \mapsto \int_{\Eset \times \Bset} \mu(\rmd (y, z)) \, P(y, \msc)$ is invariant for $Q$, and since $\pi_{\msc}$ is supposed to be the unique invariant probability measure of $Q$, there exists a constant $\zeta \geq 0$ such that for every $\Bset \in \Esigma_\Cset$, 
\begin{equation} \label{eq:invariant:zero}
    \int_{\Eset \times \Bset} \mu(\rmd (y, z)) \, P(y,\msc) = \zeta \pi_{\msc}(\Bset)\eqsp.
\end{equation} 
We show by contradiction that $\zeta > 0$. Indeed, if $\zeta=0$, then setting $\Bset = \msc$ in \eqref{eq:invariant:zero} yields
\begin{equation} \label{eq:invariant:one}
    \int_{\mse\times \msc} \mu(\rmd (y, z)) \, P(y, \msc) = 0\eqsp. 
\end{equation}
Consequently, for every $\Aset \in \Esigma$,  
\begin{align*}
  &    \mu(\Aset \times \msc) = \mu R (\Aset \times \msc)\\
  &= \int_{\mse\times \msc} \mu(\rmd (y, z)) \, P(y, \msc^c \cap \Aset) + \int_{\mse\times \msc} \mu(\rmd (y, z)) \, P(y,\msc) Q(z,\msc \cap \Aset) \\ &= \int_{\mse\times \msc} \mu(\rmd (y, z)) \, P(y,\msc^c \cap \Aset) + 0 \\ 
    &= \int_{\mse\times \msc} \mu(\rmd (y, z)) \, P(y,\msc^c \cap \Aset) + \int_{\mse\times \msc} \mu(\rmd (y, z)) \, P(y,\msc \cap \Aset) \\
    &= \int_{\mse\times \msc} \mu(\rmd (y, z)) \, P(y, \Aset)\eqsp. 
\end{align*}
Therefore, the probability measure $\Esigma \ni \Aset \mapsto \mu(\Aset \times \msc)$ is invariant for $P$, and since $\pi$ is the unique $P$-invariant probability measure, it holds that $\mu(\cdot \times \msc) = \pi$. This allows us to rewrite \eqref{eq:invariant:one} according to 
$$
    0 = \int_{\mse\times \msc} \mu(\rmd (y, z)) \, P(y, \msc) = \int_{\mse\times \msc} \pi(\rmd y) \, P(y, \msc) = \pi(\msc), 
$$
which contradicts \Cref{rem:positive}. Thus, \eqref{eq:invariant:zero} holds with $\zeta > 0$. Consequently, for every $h \in \posfunc{\Esigma \tensprod \Csigma}$, 
\begin{align}
 &   \int_{\msc \times \msc} \mu(\rmd (y, z)) \, h(y, z) = \int_{\mse \times \msc} \mu(\rmd (y, z)) \, R(h \indi{\msc \times \msc})(y, z) \nonumber \\
    & \qquad \qquad = \int_{\mse\times \msc} \mu(\rmd (y, z)) \, P(y,\msc) \int_{\msc}  Q(z, \rmd z') \, h(z', z')  = \zeta \int \pi_{\msc}(\rmd z) \, h(z, z)\eqsp. \label{eq:uniq:one}
\end{align}
The probability measure $\mu$ is invariant for $R$, but since $\Cset \times \Cset$ is not necessarily $\mu$-accessible we cannot use Kac's theorem directly. Instead, \Cref{rem:arbitrary} shows that for every bounded function $h \in\posfunc{\Esigma \tensprod \Csigma}$,
\begin{align*}
    \mu(h) &\geq \int_{\msc \times \msc} \mu(\rmd (y, z)) \, \PE_{(y, z)}^R \lrb{\sum_{k=0}^{\sigma_{\msc\times \msc} - 1} h(Y_k, Z_k)} \\
    &=  \zeta \int \pi_{\msc}(\rmd z) \, \PE_z^P \lrb{\sum_{k=0}^{\sigma_{\msc} - 1} h(X_k, z)} = \frac{\zeta\check{\pi}(h)}{\pi(\msc)}\eqsp,
\end{align*}
where the penultimate and last equalities follow from \eqref{eq:uniq:one} and \eqref{eq:df_tilde_pi_kac_c}, respectively. Since $\zeta>0$, we have shown \eqref{eq:proof_domination_invariant_R}.  The proof is completed. 
\end{proof}

The fact that $\check{\pi}$ is a unique invariant probability measure for $R$ allows us to apply \cite[Theorem~5.2.6]{douc:moulines:priouret:soulier:2018} and Birkhoff's theorem to establish that for every measurable function $f : \Eset \times \Cset \to \rset$ such that $\check{\pi}(\abs{f}) < \plusinfty$,
\begin{equation}
  \label{eq:lln_p_bar_pi}
  \lim_{n \to \plusinfty}  n^{-1} \sum_{k=0}^{n-1}  f(Y_k, Z_k)  = \check{\pi}(f), \qquad \PP^R_{\check{\pi}}\mbox{-a.s.}
\end{equation}
However, this property is not fully satisfying since it only holds under stationarity, \ie, under $\PP^R_{\check{\pi}}$. From a simulation perspective, it is crucial that the same result holds irrespective of how the chain is initialized. Our next result addresses this question and establishes that  \eqref{eq:lln_p_bar_pi} holds for an arbitrary initial distribution, \ie, when $\PP^R_{\check{\pi}}$ is replaced by $\PP^R_{\xi}$, where $\xi \in \measureset_1(\Esigma \tensprod \Csigma)$ is arbitrary. Interestingly, the assumptions under which we operate when deriving this result also turn out to be necessary and, therefore, impossible to weaken. 
\begin{theorem}
    \label{thm:LLN_Markov}
    Assume \ref{ass:P:inv} and \ref{ass:Q:inv}. In addition, suppose that the Markov kernel $R$ admits a unique invariant probability measure $\check{\pi}$. Then the two conditions 
    \begin{enumerate}[(a)]
        \item \label{assum:LLN_Markov:a} for every $x \in \Eset$, $\PP^P_x(\rti_{\Cset} < \plusinfty) = 1$ and 
        \item \label{assum:LLN_Markov:b} for every $x \in \Cset$ and bounded measurable function $h: \msc \to \rset$,
        \begin{equation*}
            \lim_{\ell \to \plusinfty}  \ell^{-1} \sum_{k=0}^{\ell-1} h(X_k) = \pi_{\msc}(h) \eqsp, \qquad \PP^Q_{x}\mbox{-a.s.}, 
        \end{equation*}
    \end{enumerate}
    hold if and only if for every $\xi \in \measureset_{1}(\Esigma \tensprod \Csigma)$ and measurable function $f : \Eset \times \Cset \to \rset$ such that $\check{\pi}(\abs{f}) < \plusinfty$, 
    \begin{equation}
    \label{eq:lln}
    \lim_{n \to \plusinfty}  n^{-1} \sum_{k=0}^{n-1}  f(Y_k,Z_k)  = \check{\pi}(f) \eqsp, \qquad \PP^R_{\xi}\mbox{-a.s.}  
    \end{equation} 
  \end{theorem}
  It is worthwhile to note that \Cref{thm:LLN_Markov} does not presuppose a law of large numbers for the Markov kernel $P$. The proof of \Cref{thm:LLN_Markov} relies  on the following necessary and sufficient conditions---based on properties of  \emph{harmonic functions} for $P$ (\ie~measurable functions $h$ for which $P|h| < \infty$ and $Ph = h$)---for a Markov chain to satisfy a law of large numbers. 

\begin{proposition}
\label{propo:harmonic_LLN}
Let $P$ be a Markov kernel on $ \Eset \times \Esigma$ admitting an invariant probability measure $\pi$. Then the two following conditions are equivalent.
\begin{enumerate}[label=(\roman*)]
\item \label{propo:harmonic_LLN_i} For every $\xi \in \measureset_{1}(\Esigma)$ and measurable function $g: \Eset \to \rset$ such that $\pi(\abs{g}) < \plusinfty$,
\begin{equation*}
\lim_{n\to \plusinfty} n^{-1} \sum_{k=0}^{n-1} g(X_k) = \pi(g), \qquad \PP^P_{\xi}\mbox{-a.s.} 
\end{equation*}
\item \label{propo:harmonic_LLN_ii} Every bounded harmonic function $h : \Eset \to \rset$ for $P$ is constant. 
 \end{enumerate}
\end{proposition}

Since the equivalence provided by \Cref{propo:harmonic_LLN} has not, as far as we know, been established in the literature before, we provide a complete proof of this result in \Cref{sec:harmonic}. We now have all the tools needed for proving \Cref{thm:LLN_Markov}. 

\begin{proof}[Proof of \Cref{thm:LLN_Markov}]
    We aim to establish the equivalence between the conditions \ref{assum:LLN_Markov:a} and \ref{assum:LLN_Markov:b} and the law of large numbers \eqref{eq:lln}. We first assume that \ref{assum:LLN_Markov:a} and \ref{assum:LLN_Markov:b} hold. Since $\check{\pi}$ is an invariant probability measure for $R$, \Cref{propo:harmonic_LLN} applies to $R$. Therefore, given an arbitrary harmonic function  $h$ for $R$, it is sufficient to establish that $h$ is constant. With this aim in mind, we prove that for every $(y, z) \in \Eset \times \Cset$,
    \begin{equation} \label{eq:2_proof_LLN}
        h(y,z) = \int_{\msc} Q(z,\rmd z') \, h(z', z').
    \end{equation}
      Indeed, if \eqref{eq:2_proof_LLN} is true, then $h(y, z)$ does not depend on $y$ (since the right-hand side of \eqref{eq:2_proof_LLN} does not depend on $y$) and then $g(z) \eqdef h(z,z)$ is a bounded harmonic function for $Q$ since \eqref{eq:2_proof_LLN} yields $g=Qg$. Finally, combining \ref{assum:LLN_Markov:b} with \Cref{propo:harmonic_LLN} applied to $Q$ shows that $g$ is a constant function. Thus, also $h$ is constant. 
      
    We now turn to the proof of \eqref{eq:2_proof_LLN}. Set $\Aset \eqdef \{\rti_{\bmsc} < \plusinfty\}$ where $\bmsc = \Cset \times \Cset$. Since $Y_{\rti_{\bmsc}} = Z_{\rti_{\bmsc}}$ on $\Aset$, combining \Cref{lem:tech}\ref{item:tech:two} with \ref{assum:LLN_Markov:a} implies  
    \begin{multline} \label{eq:3_proof_LLN}
      \PE^R_{(y,z)}\left[ h(Y_{\rti_{\bmsc}}, Z_{\rti_{\bmsc}}) \1_{\Aset} \right] = \PE^R_{(y,z)}\left[ h(Z_{\rti_{\bmsc}}, Z_{\rti_{\bmsc}}) \1_{\Aset} \right] \\
                                                            = \PP^P_y(\rti_\Cset < \infty) \int_{\msc} Q(z, \rmd z') \, h(z',z') = \int_{\msc} Q(z,\rmd z') \, h(z', z'). 
    \end{multline}
    On the other hand, by \cite[Proposition 5.2.2(ii)]{douc:moulines:priouret:soulier:2018} there exists a random variable $W$ that is invariant for the shift operator $\shift$ and such that $\PE^R_{(y,z)}[W] = h(y,z)$ for all $(y,z) \in \Eset \times \Cset$. Thus, by the strong Markov property, since $W \circ \shift^{\rti_{\bmsc}}=W$ on $\Aset$, 
    \begin{multline*}
       \PE^R_{(y, z)} \left[ h(Y_{\rti_{\bmsc}}, Z_{\rti_{\bmsc}}) \1_{\Aset} \right] = \PE^R_{(y,z)}\lrb{\PE^R_{(Y_{\rti_{\bmsc}}, Z_{\rti_{\bmsc}})}[W] \1_{\Aset}} 
       = \PE^R_{(y,z)}[W \circ \shift^{\rti_{\bmsc}} \1_{\Aset}] \\ 
       =\PE^R_{(y, z)}[W] = h(y, z),  
    \end{multline*}
    where we used \Cref{lem:tech}\ref{item:tech:one} in the penultimate equality. Combining this equality with \eqref{eq:3_proof_LLN} yields \eqref{eq:2_proof_LLN}, and we have finally established \eqref{eq:lln}.
    
    Conversely, assuming \eqref{eq:lln}, we prove that \ref{assum:LLN_Markov:a} and \ref{assum:LLN_Markov:b} hold. Let $(y_0,z_0) \in \Eset \times \Cset$; then, applying \eqref{eq:lln} to the function $f : \Eset \times \Cset \ni (y, z) \mapsto \indi{\Cset}(y) = \indi{\bar \Cset}(y, z)$ and the measure $\xi = \delta_{(y_0, z_0)}$ yields, $\PP^R_{(y_0,z_0)}$-a.s.,  
    \begin{equation}
      \label{eq:cns:one}
      \lim_{n \to \plusinfty}  n^{-1} \sum_{k=0}^{n-1}  \indi{\Cset}(Y_k) =
      \lim_{n \to \plusinfty}  n^{-1} \sum_{k=0}^{n-1}  \indi{\bar \Cset}(Y_k,Z_k)  = \check{\pi}(\bar \Cset)=\pi(\Cset)>0, 
    \end{equation} 
    where the positivity of $\pi(\Cset)$ follows from \Cref{rem:positive} under \ref{ass:P:inv}. 
    Therefore, for all $(y_0,z_0) \in \Eset \times \Cset$, $1 = \PP^R_{(y_0, z_0)}(\rti_{\bar \Cset} <\infty) = \PP^P_{y_0}(\rti_{\Cset} < \infty)$, where the last equality follows from \Cref{lem:tech} by letting $v \equiv 1$ and $n \to \infty$ in \eqref{eq:induction}. This shows \ref{assum:LLN_Markov:a}.  
    
    We now turn to \ref{assum:LLN_Markov:b}. Pick arbitrarily $h \in \bddfunc{\Csigma}$ and $x \in \Cset$. Applying \eqref{eq:lln} to the function $f : \Eset \times \Cset \ni (y, z) \mapsto \indi{\Cset}(y)h(y)$ and the measure $\xi=\delta_{(x, x)}$ yields 
    \begin{equation*}
      \lim_{n \to \plusinfty}  n^{-1} \sum_{k=0}^{n-1} \indi{\Cset}(Y_k) h(Y_k) = \pi(\indi{\Cset} h), \qquad \PP^R_{(x, x)}\mbox{-a.s.} 
    \end{equation*}  
    Combining this limit with \eqref{eq:cns:one} provides 
    \begin{equation*}
      \lim_{n \to \plusinfty}   \frac{\sum_{k=0}^{n-1}  \indi{\Cset}(Y_k) h(Y_k)}{\sum_{k=0}^{n-1}  \indi{\Cset}(Y_k) } = \frac{\pi(\indi{\Cset} h)}{\pi(\Cset)} = \pi_\Cset(h), \qquad \PP^R_{(x, x)}\mbox{-a.s.}  
    \end{equation*}  
    Therefore, $\lim_{\ell \to \plusinfty}  \ell^{-1} \sum_{k=1}^{\ell}  h(Y_{\rti_\bCset^k}) =\pi_\Cset(h)$, $\PP_{(x, x)}^R$-a.s. However, by \Cref{lem:tech}\ref{eq:tech:four}, the embedded process $\seq{Y_{\rti_\bCset^k}}{k \in \nsetpos}$  is, under $\PP_{(x, x)}^R$, a Markov chain with transition kernel $Q$ and initial distribution $\delta_x$. Thus,       
    \begin{equation*}
      1=\PP_{(x, x)}^R \lr{\lim_{\ell \to \plusinfty}  \ell^{-1} \sum_{k = 1}^{\ell}  h(Y_{\rti_\bCset^k}) =\pi_\Cset(h)} = \PP^Q_x \lr{\lim_{\ell\to \plusinfty} \ell^{-1} \sum_{k=0}^{\ell-1} h(X_k) = \pi_{\msc}(h)},   
    \end{equation*}
    which establishes \ref{assum:LLN_Markov:b}. This completes the proof.   
\end{proof}

\subsection{Geometric ergodicity of the KKT process}
\label{sec:geom-ergod-kkt}

We introduce the following additional conditions under which we will establish that $R$ is geometrically ergodic. 

\begin{hyp}{A}
\item \label{ass:P:geom}
There exist $(\lambda,b) \in (0, 1) \times \rsetps$ and a measurable function $V_P: \Eset \to [1,\plusinfty)$  such that
\begin{equation} \label{eq:P:drift}
P V_P \leq \lambda V_P+ b \indi{\Cset} \quad \mbox{and} \quad \sup_{y \in \Cset} V_P(y) < \plusinfty. 
\end{equation}
\end{hyp}
Assumption~\ref{ass:P:geom}, which is a typical geometric drift condition, implies that $\msc$ is accessible for $P$; indeed, under \ref{ass:P:geom}, by \cite[Proposition~4.3.3]{douc:moulines:priouret:soulier:2018}, $\PP_x^P(\rti_{\msc}< \plusinfty) = 1$ for every $x \in \Eset$. 

We now provide a sufficient condition for checking \ref{ass:P:geom}. 

\begin{lemma}
    \label{lem:lem:util}
    Let $ \msc \in\Esigma $ and assume that $ \eta \eqdef  \inf_{y\in\msc^c} P(y, \msc) >0$. Then \ref{ass:P:geom} holds. 
\end{lemma}
\begin{proof}
    Pick $\gamma>1$ and define $ V_\gamma(x)= \gamma \indi{\msc^c}(x)+ \indi{\msc}(x)$, $x \in \Eset$. Then, 
    \begin{align*}
        PV_\gamma(x)&=\gamma + (1-\gamma)P(x,\msc) \\
        &\leq \lr{\gamma + \eta(1-\gamma)} \indi{\msc^c}(x)+ \gamma \indi{\msc}(x)\\
            & \leq \lambda_\gamma V_\gamma(x)+ b_\gamma \indi{\msc}(x), 
         \end{align*}
where we have set $ \lambda_\gamma=1+\eta(1/\gamma -1) \in (0,1)$ and $b_\gamma=\gamma\in \rset^*_+$. Thus, \ref{ass:P:geom} holds with $\lambda=\lambda_\gamma$, $b=b_\gamma$ and $V_P=V_\gamma$ (the latter being bounded). 
\end{proof}

\begin{remark} \ \\[-4mm]
  \begin{enumerate}[(1)] 
  \item 
  Consider the case where $\Eset\subset \rset^d$ and $P$ is an MH kernel with proposal
    kernel $\propMH$. In addition, suppose that the target $\pi$ has a density, denoted by the same symbol $\pi$, and that $\propMH$ has a transition density $r$, both with respect to the Lebesgue measure. Then 
    $$
    \inf_{y\in\msc^c} P(y, \msc) \geq \int_{\msc} \ \inf_{y \in \msc^c} \lrb{r(y,z) \min \lr{\frac{\pi(z)r(z,y)}{\pi(y)r(y,z)},1}} \rmd z 
    $$
    which is strictly positive when $\msc^c$ is compact and $\pi$ and $r$ are continuous and positive. Hence, under the latter conditions, \Cref{lem:lem:util} applies and \ref{ass:P:geom} holds true. 
  \item  More generally, when $\Eset$ is a metric space, while
    \ref{ass:P:geom} is a standard assumption if $\msc$ is a compact
    set, the condition on $\eta$ stated in \Cref{lem:lem:util} allows
    \ref{ass:P:geom} to be checked under mild assumptions on $P$ in
    the case where $\msc^{c}$ is compact (and, consequently, $\msc$ is
    non-compact). Indeed, if $P$ is Feller (\ie, for any bounded and
    continuous function $f$, $Pf$ is also bounded and continuous), the
    condition $\inf_{y\in\msc^c} P(y, \msc) >0$ holds if
    $P(y,\msc) >0$ for every $y\in\msc^c$. Indeed, when $\msc^c$ is
    compact, the set $\msc$ is open and \cite[Proposition
    12.1.8]{douc:moulines:priouret:soulier:2018} implies that
    $y \mapsto P(y,\msc)$ is lower semicontinuous and that there
    exists $y_\star$ in the compact set $\msc^c$ such that
    $\inf_{y\in\msc^c} P(y, \msc)=P(y_\star,\msc)>0$. Therefore,
    \Cref{lem:lem:util} applies, and hence \ref{ass:P:geom} holds
    true. Using this approach, we can easily check that the condition
    of \Cref{lem:lem:util} (and hence \ref{ass:P:geom}) is satisfied
    for all the examples studied numerically in
    \Cref{sec:numerics}. Note also that \Cref{prop:geom:mh} below
    provides examples where \ref{ass:P:geom} can still be satisfied in
    situations where neither $ \msc $ nor $ \msc^c $ is compact.
\end{enumerate}
\end{remark}




Let $\epsilon >0$ and $\nu \in \measureset_1(\Csigma)$. We say that $\Dset \in \Esigma_\msc$ is a \emph{$(1,\epsilon \nu)$-small set for $Q$} if for every $x \in \Dset$ and $\Aset \in \Esigma_\Cset$, $Q(x, \Aset) \geq \epsilon \nu(\Aset)$.
\begin{hyp}{A}
  \item \label{ass:Q:geom} \ \\
  \begin{enumerate}[(i)]
  \item   \label{ass:geom:D} There exists an accessible $(1,\epsilon \nu)$-small set $\Dset \in \Csigma$ for $Q$ such that $\nu(\Dset)>0$.
  \item \label{ass:geom:mino:C} It holds that $ \inf_{z \in \Dset} P(z, \Cset) > 0$.
  \item \label{ass:geom:drift} There exist constants $(\lambda,b) \in (0,1) \times \rsetps$ and a measurable function $V_Q: \Cset \to [1,\plusinfty)$ such that
  \begin{equation} \label{eq:Q:drift}
  Q V_Q\leq \lambda V_Q+ b \indi{\Dset}\quad \mbox{and} \quad \sup_{z \in \Dset} V_Q(z)<\infty\eqsp.
  \end{equation}
  \end{enumerate}
\end{hyp}

Since we will assume \ref{ass:P:geom} and \ref{ass:Q:geom} jointly, we can always, for simplicity, assume that \eqref{eq:P:drift} and \eqref{eq:Q:drift} hold for the \emph{same} constants $(\lambda, b)$ (for instance, the largest ones). Note that \ref{ass:Q:geom}\ref{ass:geom:D} implies (see \cite[Theorem~9.2.2]{douc:moulines:priouret:soulier:2018}) that $Q$ is
irreducible, and therefore $\pi_{\Cset}$ is the unique invariant probability measure of $Q$ under
\ref{ass:Q:inv}. Moreover, \ref{ass:Q:geom}\ref{ass:geom:mino:C} holds true in the case where $\Eset$ is a topological space (equipped with its Borel $\sigma$-field $\Esigma$) and there exists a kernel $T$ on $\Eset$ and a compact set $\Dset$ such that $P(y, \Cset) \geq T(y, \Cset) > 0$ for all $y \in \Dset$ and the mapping $y \mapsto T(y,\Cset)$ is lower semi-continuous. This setting covers, for instance, the case where $P$ is the Markov kernel associated with the Metropolis--Hastings algorithm with positive proposal transition density. 
 
\Cref{thm:ergo:geom} below states that the KKT process is geometrically ergodic provided that $Q$ is geometrically ergodic on $\Cset$ and $P$ \emph{pushes toward} $\Cset$ according to the geometric drift condition \ref{ass:P:geom}. Therefore, remarkably, no assumption concerning smallness of $\Cset$ with respect to $P$ is made, and the set $\Cset$ may hence be potentially \emph{large}. Note also that no assumption on geometric ergodicity of the Markov kernel $P$ is needed. \Cref{thm:ergo:geom} conveys the idea that if a Markov chain evolving according to $P$ exhibits poor ergodic behavior on $\Cset$, then the action of changing its dynamics in the same region to one governed by $Q$, where $Q$ is geometrically ergodic on $\Cset$, may result in a geometrically ergodic KKT process.  
\begin{theorem}
\label{thm:ergo:geom}
Assume \ref{ass:P:inv}--\ref{ass:Q:geom}. 
Then there exist constants $C > 0$ and $\rho \in (0, 1)$ such that for every $\Xinit \in \measureset_1(\Esigma \tensprod \Csigma)$ and $n \in \nset$,
\begin{equation} \label{eq:ergo:geom}
\tvnorm{\Xinit R^n-\check{\pi}} \leq C \rho^n \int_{\mse\times \msc} V_P(y)  V_Q(z)  \, \Xinit(\rmd(y,z)).
\end{equation}
\end{theorem}

Our approach for establishing the geometric ergodicity in \Cref{thm:ergo:geom} is non-standard. Indeed, even though it can be easily seen that $\Dset \times \Dset$ is a small set for $R$, it is not obvious (and this remains an open question) how to construct, under \ref{ass:P:geom}--\ref{ass:Q:geom}, a drift function for $R$ outside $\Dset \times \Dset$  by combining the drift functions $V_P$ and $V_Q$ associated with $P$ and $Q$, respectively. Instead, the geometrically ergodic bound \eqref{eq:ergo:geom} is obtained through a delicate application of \cite[Theorem~11.4.2]{douc:moulines:priouret:soulier:2018}.  

\begin{proof}[Proof of \Cref{thm:ergo:geom}]
    Set $\bCset \eqdef \Cset \times \Cset$ and $\bDset \eqdef \Dset \times \Dset$. Under \ref{ass:P:geom}, \cite[Proposition 4.3.3(ii)]{douc:moulines:priouret:soulier:2018} allows us to write, for every $y \in \Eset$,
    \begin{equation} \label{eq:bound:return:P}
        \PEE{P}{y}[\lambda^{-\rti_\Cset}]  \leq V_P(y) + b / \lambda < \plusinfty\eqsp. 
    \end{equation}
    Therefore, $\PPP{P}{y}(\rti_\Cset<\plusinfty)=1$ for every $y \in \Eset$, and by \Cref{lem:tech}\ref{item:tech:one}, for every $\ell \in \nsetpos$ and $(y, z) \in \Eset \times \Cset$, 
    \begin{equation} \label{eq:finite:visits}
        \PPP{R}{(y, z)}(\rti^{\ell}_{\bCset} < \plusinfty) = 1. 
    \end{equation}
    Moreover, by \Cref{lem:smallSet} and \eqref{eq:bound:return:P} again, $\bDset=\Dset \times \Dset$ is an accessible $(1,\delta \epsilon \bar \nu)$-small set for $R$, where $\delta \eqdef \inf_{z \in \Dset} P(z, \Cset)$ and $\bar \nu$ is the probability measure $\bar \nu : \Esigma \tensprod \Csigma \ni \Aset \mapsto \int_\Cset \nu(\rmd z) \, \indi{\Aset}(z, z)$ which satisfies $\bar \nu(\bDset) > 0$. Thus, in order to apply \cite[Theorem 11.4.2]{douc:moulines:priouret:soulier:2018}, we only need to show that there exists some $\gamma > 1$ such that
    \begin{equation} \label{eq:geom:ergod:zero}
        \sup_{(y,z) \in \bDset} \PEE{R}{(y, z)}[\gamma^{\rti_\bDset}] < \plusinfty\eqsp.
    \end{equation}
    
    In order to establish \eqref{eq:geom:ergod:zero} we first show that for every $(y, z) \in \Eset \times \Cset$, 
    \begin{equation} \label{eq:geom:ergod:zero:first:step}
        \PPP{R}{(y, z)}(\rti_\bDset < \plusinfty) = 1. 
    \end{equation}
    For this purpose, using \eqref{eq:finite:visits} we can write $\PPP{R}{(y, z)}(\rti_\bDset<\plusinfty) = \sum_{k = 1}^{\plusinfty} \PPP{R}{(y,z)}(\rti_\bDset =\rti^k_\bCset)$, where 
    \begin{equation*}
        \PPP{R}{(y,z)}(\rti_\bDset=\rti^k_\bCset) = \PEE{R}{(y,z)}\lrb{\indi{\Dset}(Z_{\rti^k_\bCset}) \prod_{\ell=1}^{k-1} \indi{\Dset^c}(Z_{\rti^\ell_\bCset})} = \PPP{Q}{z}(\rti_{\Dset} = k)\eqsp;
    \end{equation*}
    here the last equality is obtained by applying \Cref{lem:tech}\ref{item:tech:three} for $g_\ell \equiv \mathbf{1}$ for all $\ell \leq k$, $h_\ell \equiv \indi{\Dset^c}$ for all $\ell < k$, and $h_k \equiv \indi{\Dset}$. Thus, since \cite[Proposition 4.3.3(ii)]{douc:moulines:priouret:soulier:2018} implies that 
    \begin{equation}
    \label{eq:bound:return:Q}
        \PEE{Q}{z}[\lambda^{-\rti_\Dset}] \leq  V(z) + b  /\lambda<\plusinfty\eqsp, \quad z \in \Cset\eqsp, 
    \end{equation}
    it holds that $\PPP{R}{(y,z)}(\rti_\bDset<\plusinfty)=\PPP{Q}{z}(\rti_{\Dset}<\plusinfty) = 1$, which completes the proof \eqref{eq:geom:ergod:zero:first:step}. 
    
    We now turn to the proof of \eqref{eq:geom:ergod:zero}. Set $\MP \eqdef \sup_{y \in \Cset} V_P(y)+b/\lambda$ and pick $\gamma \in (1,\exp(\log^2 \lambda / \log \MP))$. Using \eqref{eq:bound:return:P},  we have 
    \begin{equation} \label{eq:def_bound_M_p}
1\leq \lambda^{-1} \leq        \sup_{y \in \Cset} \PEE{P}{y}[\lambda^{-\rti_\Cset}] \leq \MP\eqsp.
\end{equation}
Therefore, we get
    \begin{equation} \label{eq:boundGamma}
        \gamma \leq \exp(\log^2 \lambda / \log \MP) = (\lambda^{-1})^{- \log \lambda / \log \MP} \leq \lambda^{-1} \eqsp.
    \end{equation}
By Jensen's inequality and \eqref{eq:def_bound_M_p} it holds, for every $x \in \Cset$,
    \begin{multline} \label{eq:majo:lambda}
        \PEE{P}{x}[\gamma^{\rti_\Cset}] = \PEE{P}{x}[\lambda^{\rti_\Cset \log \gamma / \log \lambda}] \leq  \PEE{P}{x}[\lambda^{-\rti_\Cset }]^{- \log \gamma / \log \lambda} \\ 
        \leq \MP^{- \log \gamma / \log \lambda} 
        = \exp \left( - \frac{\log \MP \log \gamma}{\log \lambda} \right) \leq \lambda^{-1}\eqsp.
    \end{multline}
    We will soon make use of the bound \eqref{eq:majo:lambda}; however, first, let $(y, z) \in \Eset \times \Cset$ and write, using \eqref{eq:geom:ergod:zero:first:step} and \eqref{eq:finite:visits}, 
    \begin{equation} \label{eq:bound:gamma}
        \PEE{R}{(y,z)}[\gamma^{\rti_\bDset}] = \PEE{R}{(y, z)} \left[ \gamma^{\rti_\bDset} \indin{\rti_\bDset<\plusinfty} \right] = \sum_{k=1}^{\plusinfty} r_k \eqsp, 
\text{ with }        r_k \eqdef \PEE{R}{(y, z)} \left[ \gamma^{\rti^k_\bCset} \indin{\rti_\bDset=\rti^k_\bCset} \right]\eqsp.
    \end{equation}
    Applying \Cref{lem:tech}\ref{item:tech:three} with $g_{\ell}(t) \equiv \gamma^t$ for $\ell \in \{1,\ldots,k\}$ and $t \in \nset$, $h_\ell \equiv \indi{\Dset^c}$ for $\ell <k$, and $h_k \equiv \indi{\Dset}$ (and using the convention $\rti^0_\bCset = 0$) yields  
    \begin{align*}
    r_k &= \PEE{R}{(y,z)} \left[ \gamma^{\rti^k_\bCset}  \indi{\Dset}(Z_{\rti^k_\bCset}) \prod_{\ell=1}^{k-1}\indi{\Dset^c}(Z_{\rti^\ell_\bCset}) \right] \\
&= \PEE{R}{(y,z)}\lrb{\lr{\prod_{\ell=1}^k \gamma^{\rti^\ell_\bCset-\rti^{\ell-1}_\bCset}} \indi{\Dset}(Z_{\rti^k_\bCset}) \prod_{\ell=1}^{k-1} \indi{\Dset^c}(Z_{\rti^\ell_\bCset})} \nonumber \\
    &= \PEE{P}{y}[\gamma^{\rti_{\Cset}}] \, \PEE{Q}{z}\lrb{\indi{\Dset}(X_k) \prod_{\ell=1}^{k-1} \indi{\Dset^c}(X_\ell) \, \PEE{P}{X_\ell}[\gamma^{\rti_\Cset}]} \\
    &\leq \lr{\sup_{x \in \Cset}\PEE{P}{x}[\gamma^{\rti_\Cset}]}^{k - 1} \PEE{P}{y}[\gamma^{\rti_\Cset}] \, \PPP{Q}{z}(\rti_{\Dset} = k)\eqsp.
    \end{align*}
    Now, combining this bound with \eqref{eq:boundGamma} and \eqref{eq:majo:lambda} yields for any $(y,z) \in \mse \times \msc$,
    \begin{equation*}
        r_k \leq \lambda^{- k + 1} \PEE{P}{y}[\lambda^{-\rti_\Cset}] \, \PPP{Q}{z}(\rti_{\Dset} = k)\eqsp.
    \end{equation*}
Plugging the previous inequality into \eqref{eq:bound:gamma} and using \eqref{eq:bound:return:P} and \eqref{eq:bound:return:Q} we obtain 
    \begin{equation} \label{eq:geom:ergod:one} 
        \PEE{R}{(y,z)}[\gamma^{\rti_\bDset}] \leq \lambda \PEE{P}{y}[\lambda^{-\rti_\Cset}] \, \PEE{Q}{z}[\lambda^{-\rti_\Dset}]\leq \lambda (V_P(y) + b/\lambda)(V_Q(z) + b/\lambda)\eqsp.
    \end{equation}
    Thus, using \ref{ass:P:geom} and \ref{ass:Q:geom}\ref{ass:geom:drift}, we conclude that 
    \begin{align*}
        \textstyle\sup_{(y,z) \in \bDset} \PEE{R}{(y,z)}[\gamma^{\rti_\bDset}] \leq \lambda(\sup_{y \in\Dset} V_P(y) + b/\lambda)(\sup_{z \in \Dset} V_Q(z) + b/\lambda) < \infty\eqsp, 
    \end{align*}
    which completes the proof of \eqref{eq:geom:ergod:zero}.
    
    Finally, all the assumptions of \cite[Theorem 11.4.2]{douc:moulines:priouret:soulier:2018} are satisfied and there exist $c > 0$ and $\rho \in (0,1)$ such that for every $\Xinit \in \measureset_1(\Esigma \tensprod \Csigma)$ and $n \in \nset$,
    $$
        \tvnorm{\Xinit R^n-\check{\pi}} \leq c \, \rho^n\, \PEE{R}{\Xinit}[\gamma^{\rti_\bDset} ] \leq c \lambda \, \rho^n \, \int \Xinit(\rmd (y, z))(V_P(y) + b/\lambda)(V_Q(z) + b/\lambda)\eqsp,
    $$
    where we used \eqref{eq:geom:ergod:one} in the last step. Combining this with the bound $(V_P(y)+b/\lambda)(V_Q(z)+b/\lambda) \leq (1+b/\lambda)^2V_P(y)V_Q(z)$ and letting $C=c \lambda (1+b/\lambda)^2$ completes the proof of \eqref{eq:ergo:geom}.
\end{proof}


\section{A general Kick-Kac teleportation process}
\label{sec:extensions}
\subsection{A generalization of Kac's theorem}
Recall that in order to construct the Markov teleportation process, we combined the kernel $P$, having invariant probability measure $\pi$, with an auxiliary kernel $Q$, having $\pi_\Cset$ as invariant probability measure. We now generalize this construction by replacing $\pi_\msc$ by a more general probability measure $\tilde{\pi} \in \measureset_1(\Esigma)$ which is assumed to admit a bounded density with respect to $\pi$; in other words, by assumption there exists $M \geq 1$ such that for every $x \in \mse$,
\begin{equation} \label{eq:cond:reject}
  \frac{\rmd \tilde \pi}{\rmd \pi}(x) \leq M. 
\end{equation}
Under this assumption we define the $[0, 1]$-valued measurable function 
\begin{equation} \label{eq:tildepi} 
 \alpha(x) \eqdef \frac{1}{M} \frac{\rmd \tilde\pi}{\rmd \pi}(x) \eqsp, \quad x \in \Eset. 
\end{equation}
 Of course, $\tilde \pi=\pi_\msc$
satisfies \eqref{eq:cond:reject} with $M=1/\pi(\msc)$, and in that
case, $\alpha=\indi{\Cset}$.

Before introducing the general Markov teleportation process, we first
provide a generalization of Kac's theorem. Let $\bar{\Eset} \eqdef \Eset \times [0,1]$ and $\bar{\Esigma} \eqdef \Esigma \tensprod {\mathcal B}([0,1])$ and define the Markov kernel $\bar{P}$ on $\bar{\Eset} \times \bar{\Esigma}$ by 
\begin{equation} \label{eq:def:barP}
    \bar P(\bar x,\msa) \eqdef \int P(x,\rmd x') \, \indi{\lrb{0,1}}(u') \indi{\msa}(x',u') \, \rmd u', \quad \bar{x} = (x, u) \in \bar{\Eset}, \ \Aset \in \bar{\Esigma}. 
\end{equation}

Note that a transition according to $\bar P$ can be decomposed into two
independent moves, where the first component is updated according to the
Markov kernel $P$ and the second component is drawn from the uniform distribution on
$[0, 1]$ independently of the past. For a given $\mu \in \measureset_1(\bar{\Esigma})$, denote by $\bar \PP_\Xinit$ (and $\bar \PE_\Xinit$) the unique probability distribution (and associated expectation) on the canonical space $(\bar{\Eset}^{\nset}, \bar{\Esigma}^{\tensprod \nset})$ induced by the Markov kernel $\bar P$ and the initial distribution $\Xinit$. By abuse of notation, if $\mu=\delta_{(x,u)}$ for some $(x,u) \in\Eset \times [0,1]$, we simply write $\bar \PP_{x,u}$ and $\bar \PE_{x,u}$ instead of $\PP_{\delta_{(x,u)}}$ and $\PE_{\delta_{(x,u)}}$, respectively. Moreover, when $\varphi$ is a nonnegative or bounded measurable function on $(\Eset^\nset,\Esigma^{\tensprod \nset})$ such that $\bar \PE_{x,u} \lrb{\varphi}$ does not depend on $u \in \lrb{0,1}$, we indicate this by writing $\bar \PE_{x,*} \lrb{\varphi}$. Finally, set $\bCset \eqdef \set{(x,u) \in \bar\Eset}{u \leq \alpha(x)}$ and define
\begin{equation} \label{eq:def:barC}
  \sigma_\bCset \eqdef \inf\set{k\geq 1}{(X_k,U_k) \in \bCset}\eqsp, 
\end{equation}
where $\seq{X_k,U_k}{k\in\nset}$ is the canonical process on $(\bar{\Eset}^{\nset}, \bar{\Esigma}^{\tensprod \nset})$. In addition,    define $\bar \pi \eqdef \pi \tensprod \unifDist([0,1])$ and note that $\bar{\pi}$ is an invariant probability measure for $\bar{P}$. By showing that $\bar{\msc}$ is $\bar{\pi}$-accessible and applying Kac's theorem to $\bar{P}$ we obtain the following generalization of Kac's theorem. 

\begin{proposition}
    \label{prop:kac:general}
    Let $P$ be a Markov kernel on $\Eset \times \Esigma$ with invariant probability measure $\pi$. Let $\alpha: \Eset \to \lrb{0,1}$ be a measurable function such that $\{ x \in \Eset \,  : \,  \alpha(x) >0\}$ is
    $\pi$-accessible for $P$. Then,  $\bar{\msc}$ is $\bar{\pi}$-accessible. Moreover,
    \begin{equation} \label{eq:2}
        \pi = \pi^0_\alpha = \pi^1_\alpha \eqsp,   
    \end{equation}
    where for every measurable nonnegative or bounded function $f$ on $(\Eset, \Esigma)$,
    \begin{equation*}
        \pi^0_\alpha(f) = \int \pi(\rmd x) \, \alpha(x) \bar \PE_{x,*} \lrb{\sum_{k=0}^{{ \rti}_\bCset-1}f(X_k)} \eqsp,        \pi^1_\alpha(f) = \int \pi(\rmd x) \, \alpha(x) \bar \PE_{x,*} \lrb{\sum_{k=1}^{ \rti_\bCset} f(X_k)} \eqsp. 
    \end{equation*}
\end{proposition}

\begin{proof}
 To apply Kac's theorem to the set $\bar \Cset$ for the kernel $\bar P$ we need to check that $\bar \Cset$ is $\bar\pi$-accessible for $\bar P$. Since, by assumption, $\{\alpha > 0\}$ is $\pi$-accessible for $P$, we conclude that $\sum_{k=1}^\infty \PP^P_x(\alpha (X_k) >0)>0$ for $\pi$-almost all $x\in \Eset$. Note that since $\PP^P_x(\alpha
    (X_k) >0)>0$ if and only if $\PE^P_x[\alpha (X_k)]>0$, it holds that     
    $$
    \sum_{k=1}^\infty \bar P^k((x, u), \bCset) = \sum_{k=1}^\infty  \bar
    \PE_{x,u}\lrb{\indiacc{U_k \leq \alpha (X_k)}} = \sum_{k=1}^\infty 
    \PE^P_{x}\lrb{ \alpha (X_k)} > 0
    $$
    for $\bar \pi$-almost all $(x,u) \in \bar \Eset$. Hence, $\bar{\Cset}$
    is $\bar{\pi}$-accessible for $\bar{P}$, and Kac's theorem (\Cref{thm:kac}) applies. Thus, for every $h \in \posfunc{\bar{\Esigma}}$, 
    \begin{align*}
        \bar \pi (h)&= \int_\bCset \bar \pi (\rmd (x, u)) \, \bar{\PE}_{x,u} \lrb{\sum_{k=0}^{\rti_\bCset-1}h(X_k, U_k)} \\
        &=\int \pi(\rmd x) \, \int_{[0,1]}   
        \indiacc{u \leq \alpha(x)} \bar
        \PE_{x,u}\lrb{\sum_{k=0}^{\rti_\bCset-1}h(X_k,U_k)} \rmd u \eqsp. 
    \end{align*}
    Now, if $f \in \posfunc{\Esigma}$, then setting $h(x,u)=f(x)$ yields  
    \begin{align*}
        \pi(f) = \bar \pi (h) &= \int \pi(\rmd x) \, \int_{[0,1]} \indiacc{u \leq  \alpha(x)} \bar{\PE}_{x,u}   \lrb{\sum_{k=0}^{\rti_\bCset-1}f(X_k)} \rmd u \\ 
        &= \int \pi(\rmd x) \, \alpha(x) \bar \PE_{x,*} \lrb{\sum_{k=0}^{\rti_\bCset-1}f(X_k)}  \eqsp, 
    \end{align*}
    where the last equality stems from the fact that  $\bar
    \PE_{x,u}[\sum_{k=0}^{\rti_\bCset-1}f(X_k)]$ does not depend on
    $u$. Hence, $\pi= \pi^0_\alpha$. The proof of $\pi= \pi^1_\alpha$
    follows the same lines and is omitted for brevity. 
\end{proof}

\begin{remark}
    Similarly to \Cref{cor:kac}, the generalized Kac's theorem \eqref{eq:2} also holds if $P$ admits a unique invariant probability measure $\pi$ such that $\pi(\alpha)>0$. Indeed, in that case, $\pi(\alpha>0)>0$, and by \Cref{lem:verif_C_acc_unique_invariant}, $\{\alpha >0\}$ is
    $\pi$-accessible for $P$. Therefore, \Cref{prop:kac:general} applies, implying \eqref{eq:2}. 
  \end{remark}

  If
  $\{ x \in \Eset \, : \, \alpha(x) >0\}$ is $\pi$-accessible for $P$,
  we can, on the basis of \Cref{prop:kac:general} and
  \Cref{lem:invariant_measure_kac_process_c},  define a KKT sampler using the extended Markov kernel
  $\bar{P}$ as basis kernel and a Markov kernel
  $\bar{Q}$ leaving $\bar{\pi}_{\bCset}$ invariant as teleportation
  kernel. Moreover, note that since
  $\rmd \tilde{\pi}/\rmd \pi \propto \alpha$, it is easy to verify
  that a Markov kernel $\bar{Q}$ leaving $\bar{\pi}_{\bCset}$ invariant can
  be formed by picking some $\tilde \pi$-invariant Markov kernel $Q$ and setting $\bar{Q} \eqdef \left. Q \tensprod \mathbf{Unif}(0,1)\right|_{\bCset}$, {\ie}, $\bar Q$ is the Markov kernel on $\bCset \times \bar \Esigma_\bCset$ given by 
  \begin{equation} 
    \label{eq:bar:Q}
    \bar Q(\bar x,\msa) = \int \hspace{-2mm} \int  Q(x,\rmd x') \, \frac{\indi{\lrb{0,\alpha(x')}}(u')}{\alpha(x')} \indi{\msa}(x',u') \, \rmd u' \eqsp , \quad \bar{x} = (x, u) \in \bar{\Cset}, \ \Aset \in \bar{\Esigma}_{\bCset} \eqsp.    
  \end{equation}
  Using  \Cref{prop:kac:general}, we may now construct the \emph{general Kick-Kac
  teleportation} (\gKKT) \emph{process} described in
  \Cref{alg:Markov:telep:general}, where $Q$ is, as above, some Markov kernel leaving $\tilde{\pi}$ invariant. Finally, it is worth mentioning that all the results that we have derived for the KKT process can be obtained also for the {\gKKT} process by replacing the conditions on $P$ and $Q$ by similar conditions  
  on $\bar{P}$ and $\bar{Q}$.
  
\begin{algorithm}[h]
    \caption{The general Kick-Kac teleportation (\gKKT) process}
    \label{alg:Markov:telep:general}
    \begin{algorithmic}[1]
        \State {\bf Initialization} Draw $(Y_0, Z_0)$.
        \For{$k \gets 1$ to $n$}
            \State draw $(Y^\star_k, U_k) \sim P(Y_{k-1}, \cdot) \tensprod \unifDist([0,1])$
            \If {$U_k \geq \alpha(Y^\star_k)$}
                \State  set $(Y_k, Z_k) \gets (Y^\star_k, Z_{k-1})$
            \Else
                \State draw $Z_k \sim Q(Z_{k-1},\cdot)$
                \State set $Y_k \gets Z_k$
            \EndIf
        \EndFor
    \end{algorithmic}
\end{algorithm}

\begin{remark}
 \label{rem:upToMultiplicative}
Assume that the target distribution $\pi$ is known only up to multiplicative constant, {\ie}, 
$\pi = \pi_u / \pi_u(\Eset)$, where $\pi_u$ is a known, unnormalized finite measure on $(\Eset, \Esigma)$. Moreover, assume that there exists another unnormalized measure $\tilde{\pi}_u$ on $(\Eset, \Esigma)$, dominated by $\pi_u$, and a finite constant $M_u$ such that for every $x \in \Eset$, 
$$ 
 \frac{\rmd \tilde \pi_u}{\rmd \pi_u} (x) \leq M_u \eqsp. 
$$      
Then it is straightforward to check that \eqref{eq:cond:reject} is satisfied with $\tilde \pi = \tilde \pi_u/ \tilde \pi_u(\Eset)$ and $M = M_u \pi_u(\Eset)/ \tilde \pi_u(\Eset)$. Furthermore, as the function $\alpha$ defined by \eqref{eq:tildepi} can be written as $\alpha(x)= M_u^{-1} (\rmd \tilde \pi_u/\rmd  \pi_u) (x)$, \Cref{alg:Markov:telep:general} can be implemented even if the target is known only up to a multiplicative constant. 
\end{remark}
\begin{remark} 
  \label{rem:alternative:telep:general}
  As explained above, \Cref{alg:Markov:telep:general} can be viewed as a special instance of \Cref{alg:telep}, parameterized by the base kernel $\bar{P}$, the critical region $\bar{\Cset}$ and the teleportation kernel $\bar{Q}$, where the generation of the uniformly distributed component of the teleportative move governed by $\bar{Q}$ is omitted (as this draw is never used later in the algorithm, since neither $\bar{Q}((x, u), \cdot)$ nor $\bar{P}((x, u), \cdot)$ \emph{de facto} depends on $u$). Recalling the definitions of $\bar P$ and $\bCset$ in \eqref{eq:def:barP} and \eqref{eq:def:barC}, respectively, Lines~3--4 in Algorithm~\ref{alg:Markov:telep:general} are equivalent to
   \begin{itemize}
       \item [3':] draw $(Y^\star_k, U_k) \sim \bar P((Y_{k-1}, U_{k-1}), \cdot)$
       \item [4':]\textbf{if} $(Y^\star_k,U_k) \notin \bCset$ \textbf{then}
   \end{itemize}
   which exactly correspond to Lines 3--4 in Algorithm~\ref{alg:telep} with $Y^\star_k$ replaced by
   $(Y_k^\star,U_k)$, $P$ by $\bar P$, and $\Cset$ by $\bCset$. Moreover, alternatively, Lines 3--4 in Algorithm~\ref{alg:Markov:telep:general} can be expressed as 
   \begin{itemize}
       \item [3'':] draw $Y^\star_k \sim P(Y_{k-1},\cdot)$ and, conditionally on $Y^\star_k$, $B_k \sim \mbox{\emph{Bernoulli}}(\alpha(Y^\star_k))$
       \item [4'':]\textbf{if} {$B_k=0$} \textbf{then}
   \end{itemize}
\end{remark}

\subsection{Revisiting the hybrid kernel of \cite{brockwell:kadane:2005}}
\label{sec:revis-hybr-kern}
In this section, we compare our methodology to the one proposed in \cite{brockwell:kadane:2005}.
We first recall the construction of the hybrid kernel introduced in that work and then show that this can be seen as a particular instance of the {\gKKT} process. Let $P$ be a Markov kernel on $\Eset \times \Esigma$ leaving $\pi$ invariant and let $\phi \in \measureset_1(\Esigma)$ be a so-called \emph{re-entry proposal distribution}. We assume that $\pi$ and $\phi$ have densities with respect to the same dominating measure $\nu$, and denote, by abuse of notation, these densities by the same symbols $\pi$ and $\phi$, respectively. The construction of \cite{brockwell:kadane:2005}  follows two steps: using $P$ and $\phi$, we first construct a process $\seq{Y'_k}{k \in \nset}$ on $\Eset \cup \{\mb a\}$, where $\mb a$ is an `artificial' atom; then we obtain $\seq{Y_k}{k \in \nset}$ by removing every occurrence of the state $\mb a$ from  $\seq{Y'_k}{k\in\nset}$. Letting $c > 0$, the transitions of $\seq{Y'_k}{k\in\nset}$ can be described as follows. 

{\small
\begin{center}
\begin{tabular}{|p{5,7cm}|p{5,7cm}|}
    \hline
    \centering \textbf{Case 1 :} $Y'_{k-1}=y\in \Eset$ & \centering  \textbf{Case 2 :} $Y'_{k-1}=\mb a$ \tabularnewline
    \hline
     Draw $(Y^\star_k,U_k) \sim P(Y_k,\cdot) \tensprod \unifDist{([0,1])}$; &  Draw $(Y^\star_k,V_k) \sim \phi \tensprod  \unifDist{([0,1])}$; \tabularnewline
     if $U_k  \geq \left( 1 \wedge \frac{c\phi(Y^\star_k)}{\pi(Y^\star_k)} \right)$,
     then set $Y'_k=Y^\star_k$;  & if $V_k \leq \left( 1 \wedge \frac{\pi(Y^\star_k)}{c\phi(Y^\star_k)} \right)$, then set $Y'_k=Y^\star_k$; \tabularnewline
     otherwise, set $Y'_k=\mb a$. &  otherwise, set $Y'_k=\mb a$. \tabularnewline
    \hline
  \end{tabular}
\end{center}
}
\smallskip

It is straightforwardly seen from \textbf{Case~2} that the algorithm attempts to escape from $\mb a$ through rejection sampling with proposal distribution $\phi$ and acceptance probability $1 \wedge \frac{\pi}{c\phi}$. The density (with respect to $\nu$) of the accepted candidate is therefore $\tilde{\pi}(x) \eqdef M ([c\phi(x)] \wedge \pi(x))$, $x \in \Eset$, where $M \eqdef 1 / \int [c\phi(x)] \wedge \pi(x) \, \nu(\rmd x)$. Note that $\tilde \pi(x) \leq M \pi(x)$ for every $x \in \Eset$, showing that \eqref{eq:cond:reject} is satisfied. Moreover, it can be seen from \textbf{Case~1} that if the previous state is $y \in \Eset$, a candidate drawn from $P(y,\cdot)$ is accepted with probability $1-1 \wedge \frac{c\phi}{\pi}=1-\alpha$, where $\alpha$ is defined in \eqref{eq:tildepi}. Pruning every occurence of $\mb a$ in   $\seq{Y'_k}{k\in\nset}$, we finally obtain the same transitions as in \Cref{alg:Markov:telep:general}, except that the general Markov kernel $Q$ on Line~7 is replaced by the particular distribution $\tilde \pi$. Thus, the approach of \cite{brockwell:kadane:2005} corresponds to a memoryless version of the {\gKKT} process and falls into the class of regeneration-type algorithms. In contrast, our {\gKKT}, which allows the flexibility of a general kernel $Q$ targeting the distribution $\tilde \pi$, does not require exact sampling from $\tilde \pi$.    

\subsection{Revisiting the Metropolis--Hastings algorithm}
Here we use \Cref{rem:alternative:telep:general} to cast the Metropolis--Hastings (MH) algorithm into the framework of the {\gKKT} process. Consider the MH algorithm targeting some distribution $\pi$ on $(\Eset, \Esigma)$ by means of some proposal kernel $\propMH$ on $\Eset \times \Esigma$. We  suppose that $\pi$ and $\propMH$ have a density and a transition density, denoted by $\pi$ and $r$, respectively, with respect to some $\sigma$-finite measure $\nu \in \measureset(\Esigma)$. Then, the associated MH acceptance probability is given by
\begin{equation} \label{eq:MH:acceptance:prob}
    \alphaMHD(x,y) \eqdef \lr{\frac{\pi(y)r(y,x)}{\pi(x)r(x,y)}} \wedge 1\eqsp, \quad (x,y) \in \Eset \times \Eset \eqsp. 
\end{equation}
We first recall that the transition kernel $R^{\MH}$ of the MH algorithm based on $\propMH$ is given by 
\begin{equation}
  \label{eq:def_Metropolis}
R^{\MH}(x, \Aset) = \int_{\msa} \alphaMHD(x, y) \propMH(x,\rmd y ) + \delta_x(\Aset) \int_\Eset \{1 - \alphaMHD(x,y)\} \propMH(x,\rmd y), 
\end{equation}
for $x \in \Eset$ and $\Aset \in \mce$. 

In order to show that the MCMC chain generated by this algorithm is just a special instance of the {\gKKT} process, define the probability measure 
$$
    \tilde{\pi}(h) \eqdef \frac{\int \pi(\rmd x) \, \alphaMH(x) h(x)}{\int \pi(\rmd z) \, \alphaMH(z)}, \quad h \in \posfunc{\Esigma}\eqsp,  
$$
where $\alphaMH(x) \eqdef \int \propMH(x, \rmd y) \, \alphaMHD(x,y)$, $x \in \Eset$. Note that since $0\leq \alphaMH \leq 1$, $\tilde{\pi}$ satisfies \eqref{eq:cond:reject} with $M \eqdef 1 / \int \pi(\rmd z) \, \alphaMH(z)$. Consider now a {\gKKT} process with 
\begin{enumerate}[(i)]
      \item the $\pi$-invariant Markov kernel $P$ is given by $P^{\MH}(x,\cdot)=\delta_x$, $ x \in \Eset$, and 
      \item the $\tilde \pi$-invariant Markov kernel $Q_\alpha$ is given by
      \begin{equation}
          \QMH_\alpha(x,\Aset) = \frac{\int_\Aset \propMH(x,\rmd y) \, \alphaMHD(x,y)}{\int \propMH(x,\rmd z) \, \alphaMHD(x, z)}\eqsp,\quad (x, \Aset) \in \Eset \times \Esigma\eqsp. \label{eq:Qalpha}
      \end{equation}
\end{enumerate}
With this particular choice of $P$ and $Q_\alpha$, Algorithm~\ref{alg:Markov:telep:general}, with Lines 3--4 replaced by Lines 3''--4'' according to \Cref{rem:alternative:telep:general}, turns into Algorithm~\ref{alg:Markov:telep:MH} below. 

Now, let us examine the dynamics of Algorithm~\ref{alg:Markov:telep:MH} more closely. Assume that for some $k \in \nsetpos$, $B_k=1$; then $Y_k=Z_k$, and we assign the value of
$Y_k$ to the next values $(Y_\ell)_{\ell>k}$ until the Bernoulli variable $B_\ell$, with success probability $\alphaMH(Y_{\ell-1})=\ldots=\alphaMH(Y_k)$, takes on the value $1$. In that
case, $Y_\ell$ is drawn from the distribution $\QMH_\alpha(Z_k,\cdot) = \QMH_\alpha(Y_k,\cdot)$. Equivalently, we may draw a geometrically distributed random variable $T$ with success probability
$\alphaMH(Y_k)$ and, independently, a random variable $\tilde Y \sim
\QMH_\alpha(Y_{k},\cdot)$ and then set $Y_k=Y_{k+1}=\ldots=Y_{k+T-1}$ and $Y_{k+T}=\tilde Y$. However, this is exactly the course of action of the MH algorithm, in which, starting from $Y_k$, a draw from $\propMH(Y_k, \cdot)$ is accepted with a probability given by the function $\alphaMHD(Y_k, \cdot)$. If the move is not accepted, the value of $Y_k$ is repeated. Note that this can be equivalently described as a rejection-sampling algorithm, where samples from $\QMH_\alpha(Y_k,\cdot)$ are generated by drawing candidates from $\propMH(Y_k, \cdot)$ and accepting the same with probabilities given by the function $\alphaMHD(Y_k, \cdot)$. It is well known that for rejection sampling, the number $T$ of trials before acceptance is geometrically distributed with success probability $\int \propMH(Y_k, \rmd y) \, \alphaMHD(Y_k, y) = \alphaMH(Y_k)$. Moreover, $T$ is independent of the accepted variable $\tilde{Y}$, the latter having the desired distribution $\QMH_\alpha(Y_k,\cdot)$ conditionally on acceptance. 

This finally shows that the MH algorithm is just a particular {\gKKT} process, where the  base kernel is degenerated, $P(x,\cdot) = \delta_x$, and the teleportation kernel $\QMH_\alpha$ is simulated using rejection sampling.  
\begin{algorithm}[h]
    \caption{The MH algorithm as a {\gKKT} process}
    \label{alg:Markov:telep:MH}
    \begin{algorithmic}[1]
        \State {\bf Initialization} Set $(Y_0,Z_0) \gets (x_0, x_0)$.
        \For{$k\gets 1$ to $n$}
            \State set $Y^\star_k=Y_{k-1}$ and, conditionally to $Y^\star_k$, draw $B_k \sim \mathrm{Bernoulli}(\alphaMH(Y^\star_k))$
            \If {$B_k=0$}
                \State  set $(Y_k, Z_k)\gets (Y^\star_k,Z_{k-1})$
            \Else
                \State draw $Z_k \sim \QMH_\alpha(Z_{k-1}, \cdot)$
                \State set $Y_k \gets Z_k$ 
            \EndIf
        \EndFor
    \end{algorithmic}
\end{algorithm}


Here, $\bar{\Eset} \eqdef \Eset \times [0,1]$ and $\bar{\Esigma} \eqdef \Esigma \tensprod {\mathcal B}([0,1])$  and the Markov kernels $\bar P$ and $\bar Q$ defined in \eqref{eq:def:barP} and \eqref{eq:bar:Q} write
\begin{align} 
    \bar P(\bar x,\msa) &\eqdef \int \delta_{x}(\rmd x') \, \indi{\lrb{0,1}}(u') \indi{\msa}(x',u') \, \rmd u', \quad \bar{x} = (x, u) \in \bar{\Eset}, \ \Aset \in \bar{\Esigma}, \label{eq:defP:mh}\\
    \bar Q(\bar x,\msa) &\eqdef \int Q_\alphaMH(x,\rmd x') \, \frac{\indi{\lrb{0,\alphaMH(x')}}(u')}{\alphaMH(x')} \indi{\msa}(x',u') \, \rmd u', \quad \bar{x} = (x, u) \in \bar{\Eset}, \ \Aset \in \bar{\Esigma}, \label{eq:defQ:mh}
\end{align}
where $Q_\alphaMH$ is defined in \eqref{eq:Qalpha}. It is worthwhile to note that in this case the base kernel $ \bar P $ is not even ergodic, since when $(X_k,U_k)_{k \in \nset}$ is a Markov chain with Markov kernel $\bar P$ it holds that $n^{-1} \sum_{k=1}^n h(X_k)=h(X_0)$, $\PP^{\bar P}_\xi$-a.s., for every initial distribution $\xi$ on $(\Eset,\Esigma)$. In contrast, since the MH algorithm is a particular \gKKT\ process, we obtain that the teleported process $(Y_k)_{k \in \nset}$ is ergodic  by applying \cite[Corollary 2]{tierney:1994}. Thus, 
in this particular case, using the teleportation kernel instead of the original base kernel always increases the efficiency of the algorithm.  

Next, we show that the conditions that we derive in \Cref{sec:geom-ergod-kkt} are mild in the sense that they are ``almost'' necessary conditions for the MH algorithm to be geometrically ergodic.
\begin{proposition}
    \label{prop:geom:mh}
    Assume that the following conditions hold.  
\begin{enumerate}[(I)]
    \item \label{item:pos} $\alphaMH_0\eqdef \inf_{x \in \Eset} \alphaMH(x)>0$ and for every $(x,x') \in \Eset \times \Eset$, $\pi(x)>0$ and $r(x,x')>0$.  
    \item  \label{item:geom:mh:two} There exists a set $\Dset_{\MH} \in \Esigma$, a constant $\epsilon>0$ and a probability measure $\nu_{\MH}$ on $(\Eset,\Esigma)$ such that $\Dset_{\MH}$ is $\propMH$-accessible, $\nu_{\MH}(\Dset_{\MH})>0$ and $\Dset_{\MH}$ is $1$-small for $\QMH_{\alpha}$ in the sense that for every $x \in \Dset_{\MH}$ and $\Aset \in \Esigma$,   
    $$ 
\QMH_{\alpha}(x,\msa) \geq \epsilon \nu_{\MH}(\Aset) \eqsp.  
    $$
    \item \label{item:geom:mh:three} There exist constants $(\lambda_{\MH},b_{\MH}) \in (0,1) \times \rsetps$ and a measurable function $V_{\MH}: \mse \to [1,\plusinfty)$ such that
    \begin{equation} \label{eq:R:drift}
R^{\MH} V_{\MH}(y)\leq \lambda_{\MH} V_{\MH}(y)+ b_{\MH} \indi{\Dset_{\MH}}\quad \mbox{and} \quad \sup_{x \in \Dset_{\MH}} V_{\MH}(x)<\infty\eqsp.
    \end{equation}
    Then \ref{ass:P:geom} and \ref{ass:Q:geom} hold true with $P$ and $Q$ replaced by $\bar P$ and $\bar Q$, respectively.  
\end{enumerate}
\end{proposition}

Assumption~\ref{item:pos} is necessary for $R^{\MH}$ to be geometrically ergodic by \cite[Proposition 5.1]{roberts1996geometric}.
In Assumption~\ref{item:geom:mh:two}, the fact that $\Dset_{\MH}$ is supposed to be $K^{\MH}$-accessible is mild.  By \cite[Proposition~4.3.3]{douc:moulines:priouret:soulier:2018} (see the comment after \ref{ass:P:geom}) it holds, for instance, if $K^{\MH}$ satisfies an appropriate drift condition. 
In addition, the condition that $\Dset_{\MH}$ is $1$-small for $\QMH_{\alpha}$ is in general assumed when showing that this set is  $1$-small for $R^{\MH}$. Finally, Assumption~\ref{item:geom:mh:three} holds (up to some constants) if $R^{\MH}$ is $V_{\MH}$-geometrically ergodic; see \cite[Chapter 15]{meyn:tweedie:2012}.

\begin{proof}
Let $\bCset \eqdef \set{(x,u) \in \bar\Eset}{u \leq \alphaMH(x)}$. By the definition \eqref{eq:defP:mh} of $\bar P$ it holds that $\bar P(\bar x,\bar \msc)=\alphaMH(x)\geq \alphaMH_0>0$ for every $\bar x=(x,u) \in \Eset$. Hence, \Cref{lem:lem:util} applies, which implies \ref{ass:P:geom}. 

We now check \ref{ass:Q:geom}. By Assumption~\ref{item:pos} it holds that $\alphaMHD(x,y)>0$ for all $(x, y) \in \Eset \times \Eset$, and since $\Dset_{\MH}$ is accessible for $\propMH$ we deduce that $\bar{\Dset} \eqdef \Dset_{\MH} \times [0,1]$ is accessible for $\bar Q$ defined in \eqref{eq:defQ:mh} (this can be done by verifying that $[\propMH]^k(x,\msa)=0$ whenever $\bar{Q}^k(\bar x,\msa \times [0,1]) = 0$ for $\bar x =(x,u)\in\mse \times [0,1]$, $\msa \in\mce$ and $k \in\nset$).

In addition, by \eqref{eq:defQ:mh}, \eqref{eq:Qalpha} and Assumption~\ref{item:geom:mh:two} it holds that for every $\bar x \in \bar{\Dset}$ and $(\Aset,\Bset) \in \Esigma \times {\mathcal B} ([0,1])$, as $\alphaMH_0 \leq \alphaMH(x)\leq 1$,
\begin{align*}
    \bar Q(\bar x, \Aset\times \Bset) &\geq \QMH_{\alpha}(x,\msa) \Leb(\Bset \cap [0,\alphaMH_0]) \geq \epsilon \nu(\Aset \times \Bset), 
\end{align*}
where $\Leb$ is the Lebesgue measure on $[0,1]$ and $\nu$ is the measure defined by $\nu(\Aset \times \Bset)=\nu_{\MH}(\Aset) \Leb(\Bset \cap [0,\alphaMH_0])$. Thus, $\bar{\Dset}$ is a $(1,\epsilon \nu)$-small set for $\bar Q$ with $\nu(\bar \Dset)=\nu_{\MH}(\Dset_{\MH}) \alphaMH_0>0$, which implies that \ref{ass:Q:geom}\ref{ass:geom:D} is satisfied. Moreover, for every $\bar x\in \bar{\Dset}$, 
$$ 
 \bar P(\bar x,\bCset) =\int \delta_x(\rmd x') \indi{[0,1]}(u') \, \rmd u' \ \indiacc{u'\leq \alphaMH(x')}=\alphaMH(x) \geq \alphaMH_0 > 0, 
$$ 
which implies that \ref{ass:Q:geom}\ref{ass:geom:mino:C} is satisfied. Define $V_{\bar Q}(\bar x)\eqdef V_{\MH}(x)$ for $\bar x=(x,u) \in \bar \Eset$; then by Assumption~\ref{item:geom:mh:three}, 
\begin{align*}
  \bar Q V_{\bar Q}(\bar x) &= \alphaMH(x)^{-1}\int \propMH(x, \rmd y) \, \alphaMHD(x, y) V_{\MH}(y)\\
&=  \alphaMH(x)^{-1} \lr{R^{\MH} V_{\MH}(x)-[1-\alphaMH(x)] V_{\MH}(x)} \\
& \leq \alphaMH(x)^{-1} \lr{\lrb{\lambda_{\MH}-1+\alphaMH(x)} V_{\MH}(x)+b_{\MH} \indi{\Dset_{\MH}}(x)} \\ 
&\leq \lambda_{\MH} V_{\bar Q}(\bar x)+ (\alphaMH_0)^{-1}b_{\MH} \indi{\Dset_{\MH}}(x), 
\end{align*}
where we have used in the last inequality that $\alpha^{-1}(\lambda -1+\alpha) \leq \lambda$ for every $(\alpha,\lambda) \in (0,1]\times (0,1)$ and $V_{\bar Q}(\bar x)= V_{\MH}(x)$. Noting that $\sup_{\bar x \in \bar{\Dset}} V_{\bar Q}(\bar x)=\sup_{x \in \Dset_{\MH}} V_{\MH}(x)<\infty$, we conclude that \ref{ass:Q:geom}\ref{ass:geom:drift} holds true with $V_Q$ replaced by $V_{\bar Q}$. The proof is complete.  
\end{proof}


\section{Numerical illustrations}
\label{sec:numerics}
In the following we benchmark numerically the KKT sampler against the \emph{Metropolis-adjusted Langevin algorithm} (MALA) \cite{roberts:tweedie:1996} and the HMC algorithm \cite{neal:2011,betancourt-bernoulli:2017}. In all our experiments, the target distribution $\pi$ is assumed to have a positive density on $\rset^d$ with respect to the Lebesgue measure; for ease of notation, this density will be denoted by the same symbol, $\pi$. In addition, $\log \pi$ is supposed to be continuously differentiable, and we denote its gradient by $\nabla \log \pi$.

 We will consider two different parameterizations of the MH algorithm, namely the MALA and the \emph{random-walk Metropolis} (RWM) \emph{algorithm}. In RWM, proposals are generated by means of a random walk with standard deviation $\upsigma > 0$, providing an MH algorithm with proposal transition density  
\begin{equation*}
    r_{\upsigma}(x, y) = (2 \upsigma^2 \uppi)^{-d/2} \exp \left( -\frac{1}{2 \upsigma^2} \norm{y - x}^2 \right)\eqsp. 
\end{equation*}
On the other hand, in the MALA, proposals are generated according to the Euler--Maruyama discretization of the overdamped Langevin diffusion, {\ie}, 
\begin{equation*}
r_{\gamma}(x,y) = (4\gamma\uppi)^{-d/2} \exp \left( - \frac{1}{4 \gamma} \norm{y-x-\gamma \nabla \log \pi(x)}^2 \right)\eqsp,
\end{equation*}
where $\gamma >0$ is a fixed stepsize. The Markov kernels associated with the RWM algorithm and the MALA will be denoted by $\Qrwm_\upsigma$ and $\Qmala_\gamma$, respectively. 

Finally, we consider the HMC algorithm, which is briefly reviewed in the following (we refer to \cite{neal:2011,bou:sanz:2018} for details). HMC is based on the Hamiltonian function $H(x,v) \eqdef -\log \pi(x) + \norm[2]{v}/2$, $(x, v) \in \Eset \times \rset$,  associated with the potential $-\log \pi$ of the target distribution and the corresponding ordinary differential equation (ODE) and dynamics. 
Indeed, the latter preserves the extended target distribution $\pi \tensprod \gaussLaw(0,\Idd)$, where $\gaussLaw(0,\Idd)$ is the standard $d$-dimensional Gaussian distribution. However, in most cases, integrating exactly the Hamiltonian ODE is not an option, and therefore numerical integrators are typically used instead. A very popular choice is the \emph{Verlet integrator} which, given some starting point $(x_0,v_0) \in \rset^{2d}$, time step size $\Delta_t > 0$ and number $N_{\HMC} \in \nsets$ of iterations, consists in the recursion $(x_{k+1},v_{k+1}) = \Uppsi_{\Delta_t}(x_k,v_k)$, for $k\in \{0,\ldots,N_{\HMC}-1\}$, where, letting $v_{k+1/2} = v_k +2^{-1}\Delta_t \log \pi(x_k)$,
    \begin{equation}
      \label{eq:def_Uppsi}
      \Uppsi_{\Delta_t}(x_k,v_k) \eqdef (v_{k+1/2}+2^{-1} \Delta_t \, \nabla \log \pi(x_k+\Delta_t  v_{k+1/2}), x_k + \Delta_t  v_{k+1/2})\eqsp.
    \end{equation}
    For $k\in \{1, \ldots, N_{\HMC}-1\}$, $(x_k,v_k)=\Uppsi_{\Delta_t}^{\circ k}(x_0,v_0)$, where  $\Uppsi_{\Delta_t}^{\circ 1} = \Uppsi_{\Delta_t}$ and $\Uppsi_{\Delta_t}^{\circ (k + 1)} = \Uppsi_{\Delta_t}^{\circ k} \circ \Uppsi_{\Delta_t}$ for $k \in \nsetpos$. It follows that  $\Uppsi_{\Delta_t}^{\circ k}(x_0,v_0)$ is an approximation of the Hamiltonian dynamics at time $k \Delta_t$. However, failing to leave the Hamiltonian function constant, numerical integration of the Hamiltonian ODE does not preserve the extended target in general. Still, like the continuous Hamiltonian dynamics, symplecticness and reversibility still hold. These properties allow $(x_{N_{\HMC}},v_{N_{\HMC}}) = \Uppsi_{\Delta_t}^{\circ N_{\HMC}}(x_0,v_0)$ to be used as a \emph{deterministic} proposal inside an MH algorithm with corresponding acceptance probability function
    \begin{equation}
      \label{eq:alpha_hmc}
      \alpha_{\HMC}(x_0,v_0) = 1\wedge\exp\parenthese{H(x_0,v_0) - H \circ \Uppsi_{\Delta_t}^{\circ N_{\HMC}}(x_0,v_0)}.
    \end{equation}
    Finally, to ensure that the resulting Markov kernel is irreducible and ergodic (see
    \cite{durmus:moulines:saksman:2020}), the
    starting point of the additional variable $v_0$ is refreshed, \ie, sampled from its stationary distribution $\gaussLaw(0,\Idd)$ independently of the past, at each iteration. This yields the Markov kernel
\begin{multline}
  \label{eq:def_hmc}
  R^{\HMC}_{\Delta_t,N_{\HMC}}(x, \msa) \eqdef (2\uppi)^{-2d}\int_{\rset^d} \1_{\msa}\parenthese{\proj_x \circ \Phi^{\circ(N_{\HMC})}_{\Delta_t}(x,v)}  \alpha_{\HMC}(x,v) \exp(\norm[2]{v}) \, \rmd v \\
   + \delta_x(\msa) (2\uppi)^{-2d}\int_{\rset^d} (1 -  \alpha_{\HMC}(x,v)) \exp(\norm[2]{v}) \, \rmd v, \quad x \in \rset^d, \ \msa \in \mcbb(\rset^d),
\end{multline}
where $\Uppsi_{\Delta_t}$ and $\alpha_{\HMC}$ are defined in \eqref{eq:def_Uppsi} and \eqref{eq:alpha_hmc}, respectively, and $\proj_x : \rset^{2d} \ni (x,v) \mapsto x$.

\subsection{Two multi-modal distributions}
In a first example, we consider the target distribution 
\begin{equation} \label{eq:multi_modal_one}
\pi(x) = (4 \uppi)^{-1} \left( \exp \left( - \frac{1}{2} \| x - \boldsymbol{\mu} \|^2 \right) + \exp \left( - \frac{1}{2} \|x + \boldsymbol{\mu} \|^2 \right) \right), \quad x \in \rset^2, 
\end{equation}
corresponding to a mixture of two bivariate Gaussian distributions with identity covariance matrix and means $\boldsymbol{\mu} = (10, 0)^\intercal$ and $- \boldsymbol{\mu}$. 
In this setting, we compare the MALA, with step size $\gamma = 0.1$, with the memoryless KKT sampler. The latter is parameterized by the same base kernel $P = \Qmala_\gamma$, again with $\gamma = 0.1$, and the set
\begin{equation}
    \Cset = \set{x \in \Dset}{\pi(x) \leq c q(x)},  
\end{equation}
where $c = 1.3 / \uppi$, $\Dset=\ccint{-15,15}^2$ and $q(x) = \indi{\Dset} (x)/30^2$ for $x \in \rset^2$. 
Moreover, we draw from $\pi_{\msc}$ using rejection sampling with instrumental density $q$, \ie, the uniform distribution on $\Dset$. With this choice of $q$, the KKT sampler will explore the modes of $\pi$ using the MALA and teleport across the low-probability regions far from the modes through independent sampling from $\pi_\Cset$. In addition, by the definition of $\Cset$, the rejection-sampling algorithm performs relatively well; the average number of rejections equals $70$. Note that the KKT sampler is supposed to operate using the MALA also outside $\Dset$; however, due to the large span of $\Dset$, this never occurred in our simulation. \Cref{fig:double_gaussian} provides the resulting histograms of $10^6$ samples. As clear from this plot, the MALA gets stuck in one of the modes, and adjusting the step size of the proposal does not help circumventing this problem. On the other hand, the KKT process---moving easily across the low-probability barrier between the modes---explores efficiently the full distribution.  

\begin{figure}
\includegraphics[scale=0.28]{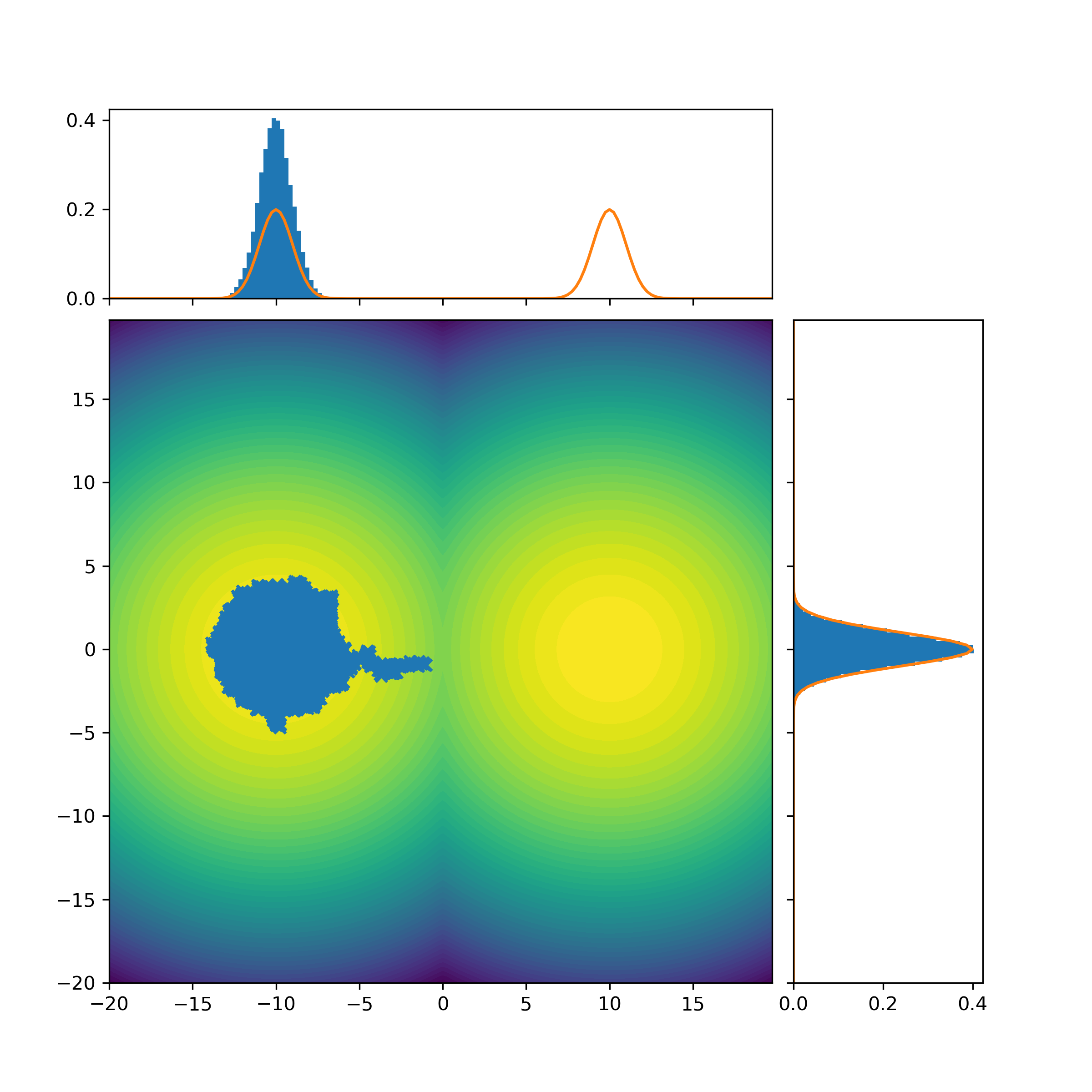}
\includegraphics[scale=0.28]{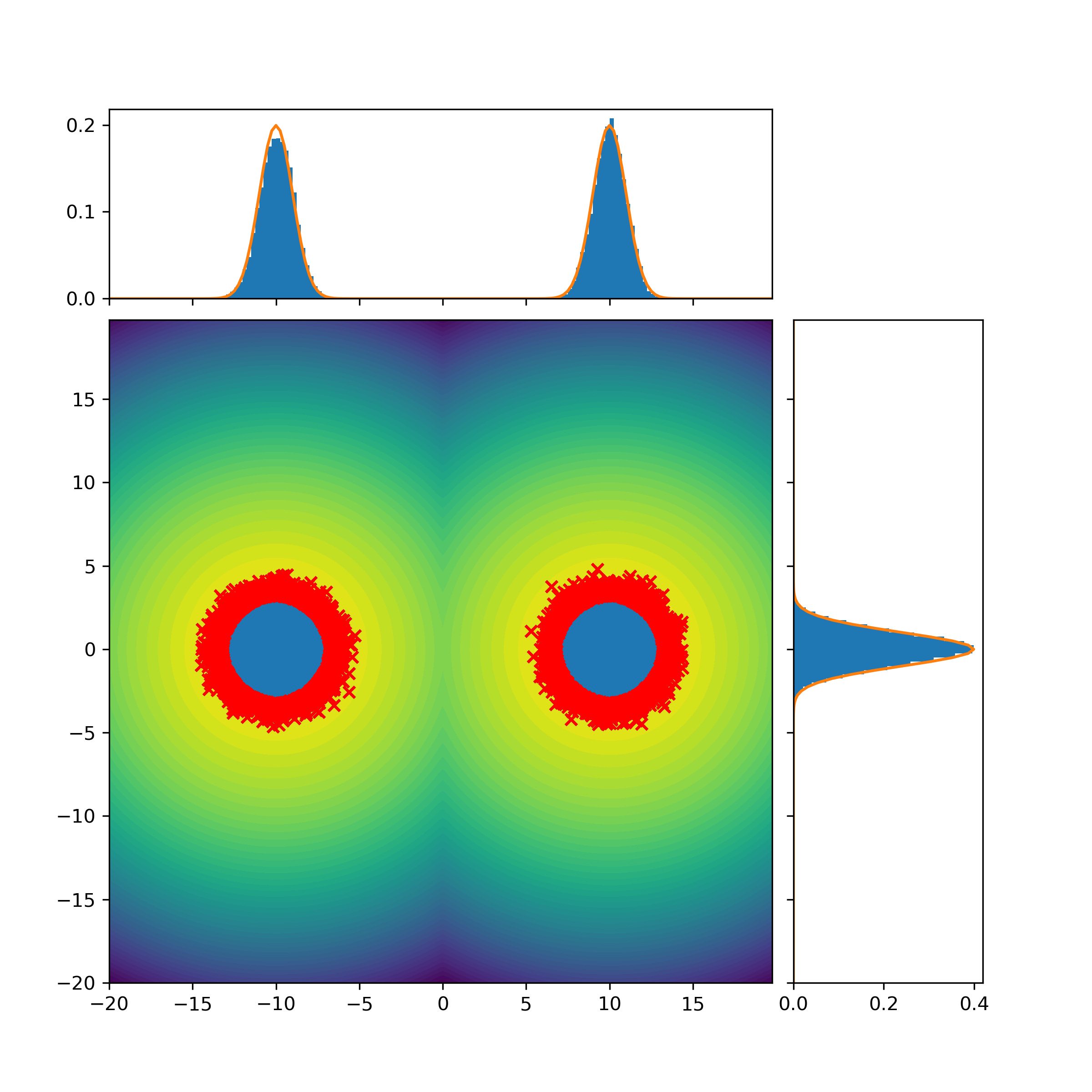}
\caption{Comparison between MALA (left panel) and the memoryless KKT sampler (right panel) on the multi-modal distribution \eqref{eq:multi_modal_one}. For the memoryless KKT sampler, red and blue points correspond to samples generated by rejection sampling from $\pi_{\msc}$ and MALA base kernel, respectively.}
  \label{fig:double_gaussian}
\end{figure}

In a second experiment we consider again a two-dimensional multi-modal distribution, this time given by  
\begin{multline} \label{eq:multi_modal_two}
  \pi(x) =14.8^{-1} \left[0.8 \left( \int \exp(-s^4) \, \rmd s \right)^{-2} \exp \left( -(x_1 + 7)^4 - (x_2 + 6.5)^4 \right) \right. \\
  \left. + (2\pi)^{-1} \sum_{i=1}^{14} \exp \left( - \frac{1}{2} \| x - \boldsymbol{\mu}_i \|^2 \right) \right], \quad x = (x_1, x_2) \in \rset^2, 
\end{multline}
where $(\boldsymbol{\mu}_i)_{i = 1}^{14}$ are points in $\rset^2$, marked with black bullets in \Cref{fig:No_grad_lips_pot}. Note that one of the components of the mixture distribution \eqref{eq:multi_modal_two} has lighter tails than the other ones, of which all are Gaussian. 

For this model, we compared the MALA, with step size $\gamma = 0.8$, with the KKT sampler parameterized by the same base kernel $P = \Qmala_{\gamma}$ (again with $\gamma = 0.8$) and the set  
\begin{equation*}
  \Cset = \{ x \in \rset^2 : - \log(14.8 \pi(x)) > 2 \}.
\end{equation*}
In this case, we let the teleportative moves of the KKT process be governed by the RWM kernel $Q = \Qrwm_{\upsigma}$ with $\upsigma = 0.8$. The algorithmic parameters $\gamma$, $\upsigma$ and $\Cset$ are tuned to obtain the best results for both methods. 
 \Cref{fig:No_grad_lips_pot} displays the resulting samples after $10^6$ iterations. As evident from the plot, MALA fails to explore the mode of the target corresponding to the component of the mixture with the lightest tails, even though this mode is less isolated than in our previous example. The KKT sampler, on the other hand, has no difficulty at all in transiting across any low-probability barrier. 
 \begin{figure}
  \includegraphics[scale=0.28]{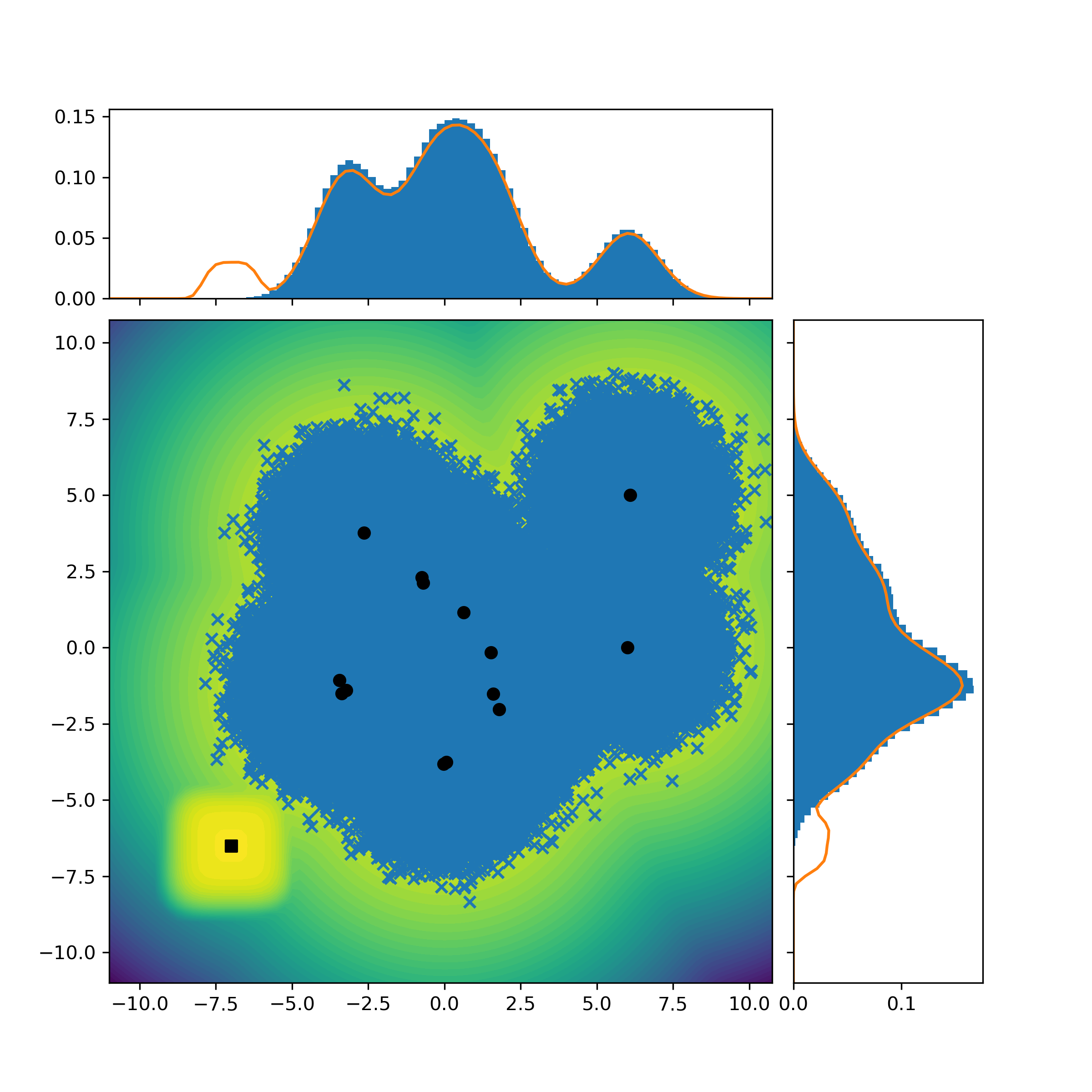}
    \includegraphics[scale=0.28]{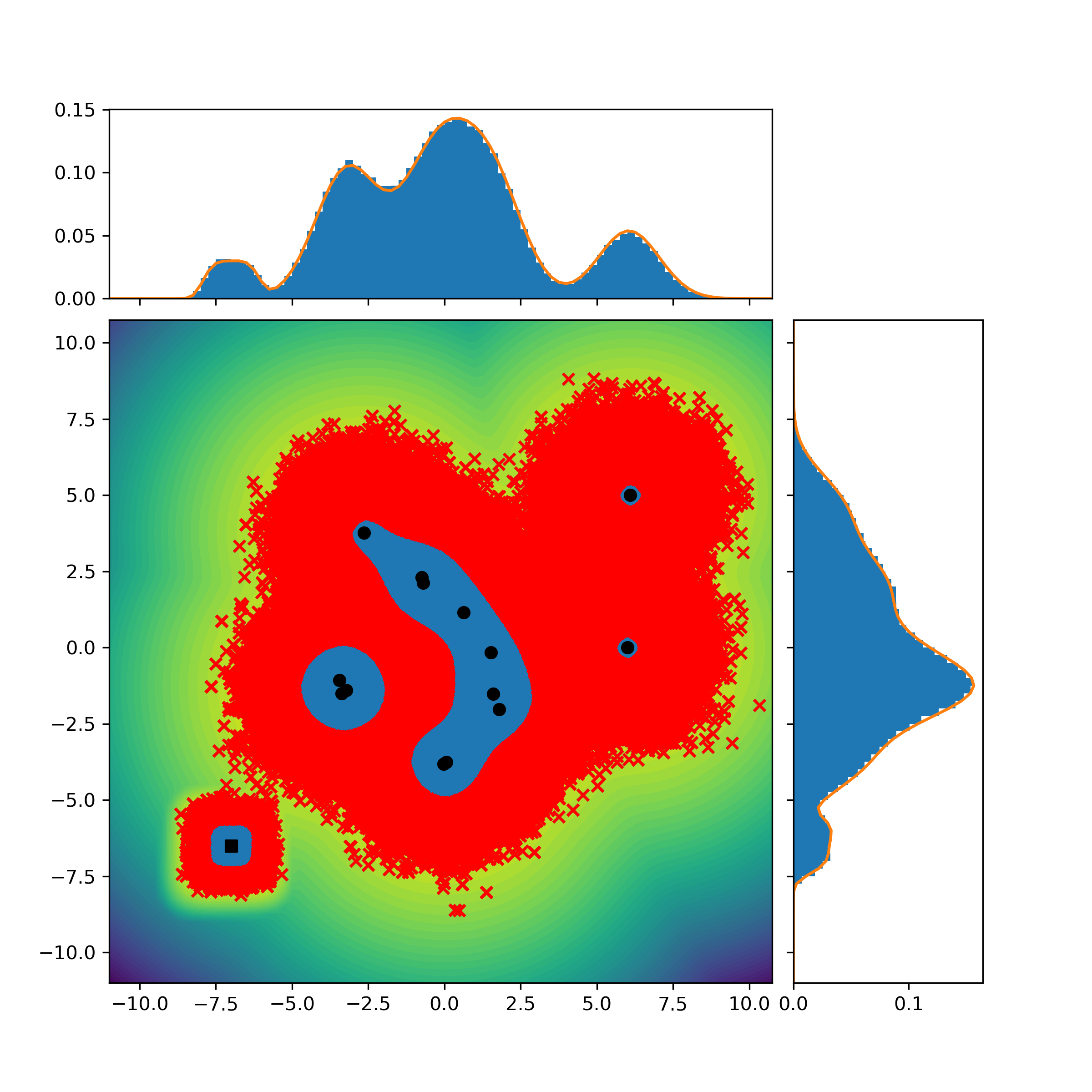}
    \caption{Comparison between MALA (left panel) and the KKT sampler (right panel) on the multi-modal distribution \eqref{eq:multi_modal_two}. For the KKT sampler, red and blue points correspond to samples generated by the RWM teleportation kernel and MALA base kernel, respectively.} 
  \label{fig:No_grad_lips_pot}
\end{figure}

\subsection{Stochastic volatility model}
\label{sec:Stochastic:Volatility}

We now illustrate the proposed KKT sampler in the context of Bayesian inference in a stochastic volatility model.  More precisely, following \cite{Ludkin2019HugAH} we assume that we have access to a given record $(y_k)_{k=0}^N \in \rset^{N + 1}$, $N \in \nsets$, of observations such that for every $k \in \{0, \ldots, N\}$, $y_k$ is a draw from a one-dimensional, zero-mean Gaussian distribution with variance $\exp(2 x_k)/\tau$, where the sequence $(x_k)_{k = 0}^N$ is governed by the recursion 
\begin{equation*}
    x_{k+1} = \rho x_k + z_{k+1}, 
\end{equation*}
with the model parameters $\rho \in \ooint{0,1}$ and $\tau > 0$ and the noise sequence $(z_k)_{k = 0}^N \in \rset^{N + 1}$ being unknown. In order to infer the parameters and the noise variables we take a Bayesian approach and assign $\tau$ a gamma prior with hyperparameter $(21, 5)$ and $(1+\rho)/2$ a beta prior with hyperparameter $(20, 2)$. Moreover, the noise variables $(z_k)_{k = 0}^N$ are assumed to be \emph{a priori} independent and standard normally distributed and $x_0 = z_0 / \sqrt{1 - \rho^2}$. In our experiment, $N + 1 = 100$ synthetic data points $(y_k)_{k = 0}^N$ were generated by simulation under model parameters and noise variables drawn from the corresponding priors. In addition, the reparameterizations 
\begin{equation*}
    \alpha = -\log(\tau)/2 \eqsp, \qquad \beta = (\log(1+\rho)-\log(1-\rho))/2.
\end{equation*}
of $\tau$ and $\rho$ provide parameters that are well supported on $\rset$. We then obtain a target density on $\rset^{N+3}$ given by 
\begin{multline*} 
-\log \pi(\alpha,\beta, (z_k)_{k=0}^N \mid (y_k)_{k=0}^N) = 42 \alpha + 5\rme^{-2\alpha}+22\log(1+\rme^{-2\beta})+4\beta+(N+1) \alpha \\
+ 2^{-1} \sum_{k=0}^{N} x_k+2^{-1}\sum_{k=0}^{N} (w_k+z_k^2)+\mathrm{Cst}_{\mathrm{SV}},
\end{multline*}
where $w_k \eqdef \exp(-x_k-2\alpha)y_k^2$ and $\mathrm{Cst}_{\mathrm{SV}}$ is some additive constant independent of the parameters of interest; see \cite[Section~G.3]{Ludkin2019HugAH} for details. On this model, we compare the KKT sampler to HMC. For HMC, we consider $N_{\HMC}= 35$ iterations for the Verlet integrator and adjust the time step $\Delta_t$ to obtain an acceptance rate of about $0.7$ on the average. The very same HMC kernel is used as base kernel $P$ in the KKT sampler. As teleportation kernel we use $Q = \Qrwm_\upsigma$, with $\upsigma$ selected to obtain an acceptance rate of about $0.25$. Finally, the region $\Cset$ is given by 
\begin{equation*}
    \Cset=\defEns{(\alpha,\beta, (z_k)_{k=0}^N) \in \rset^{N+3} : -\log \pi(\alpha,\beta, (z_k)_{k=0}^N \mid (y_k)_{k=0}^N) -\mathrm{Cst}_{\mathrm{SV}}>75}.
\end{equation*}
In this setting, we run both MCMC algorithms for $10^5$ iterations after discarding a burn-in period comprising $10^5$ steps. For the KKT sampler, the percentage of samples generated by the teleportation kernel is $63\%$. \Cref{fig:ESS_Stochastic_Volatility} reports the resulting \emph{effective sample sizes} (ESSs) (see \cite[Section 12.3.5]{Robert}) per evaluation of $\log \pi$ and $\nabla \log \pi$ at each iteration for each component of the resulting MCMC trajectories. As clear from \Cref{fig:ESS_Stochastic_Volatility}, the KKT sampler outperforms significantly HMC in terms of this measure.

\begin{table}[]
\begin{tabular}{|l|l|l|llll|}
\hline
\multicolumn{1}{|c|}{\multirow{2}{*}{Algorithm}} & \multirow{2}{*}{$\mathrm{ESS}(\alpha)$} & \multirow{2}{*}{$\mathrm{ESS}(\beta)$} & \multicolumn{4}{c|}{$(\mathrm{ESS}(z_k))_{k=0}^N$}                                                  \\ \cline{4-7} 
\multicolumn{1}{|c|}{}                           &                           &                          & \multicolumn{1}{l|}{mean} & \multicolumn{1}{l|}{variance} & \multicolumn{1}{l|}{min}  & max  \\ \hline
KKT sampler                                      & 4.84                      & 1.57                     & \multicolumn{1}{l|}{0.86} & \multicolumn{1}{l|}{0.10}     & \multicolumn{1}{l|}{0.33} & 1.89 \\ \hline
HMC                                              & 0.23                      & 0.17                     & \multicolumn{1}{l|}{0.22} & \multicolumn{1}{l|}{0.004}    & \multicolumn{1}{l|}{0.15} & 0.46 \\ \hline
\end{tabular}
\caption{Summary statistics of ESSs calculated on the basis of MCMC paths generated by the MALA and the KKT sampler for the stochastic volatility model in \Cref{sec:Stochastic:Volatility}. The values are based on $10^5$ iterations of each algorithm.}
\label{fig:ESS_Stochastic_Volatility}
\end{table}

\subsection{The Ginzburg--Landau model}
\label{sec:Ginzburg:Landau:model}

In our final numerical experiment, we consider the \emph{Ginzburg--Landau model} used for describing phase transitions in condensed matter physics; see \cite[Section~6.2]{livingstone2019kinetic}. Let $p \in \nset$ and $d = p^3$ and define $\pi$ on a three-dimensional lattice by, for $x=(x_{ijk})_{(i,j,k) \in \{1, \ldots, p\}^3} 
\in \rset^d$, 
\begin{equation*}
    - \log  \pi(x)= 2^{-1}\sum_{i,j,k=1}^p\parenthese{(1-\tau)x_{ijk}^2+\tau\alpha \| \tilde{\nabla} x_{ijk} \|^2+\tau \lambda x_{ijk}^4/2}+ \mathrm{Cst}_{\mathrm{GL}},
\end{equation*}
where $\mathrm{Cst}_{\mathrm{GL}}$ is the logarithm of the normalizing constant and $\tau$, $\lambda$ and $\alpha$ are all positive model parameters. In addition, $\tilde{\nabla} x_{ijk} \eqdef (x_{i_+jk} - x_{ijk}, x_{ij_+k} - x_{ijk}, x_{ijk_+} - x_{ijk})$, where $i_+ \eqdef i+1 \mod p$ (and similarly for $j_+$ and $k_+$). In our simulations, the target is parameterized by $p=5$, $\tau=2$, $\lambda=0.5$ and $\alpha=0.1$. In this example, we compare the MALA, operating with the step size $\gamma = 10^{-3}$, with the KKT sampler using the base kernel $P=\Qmala_{\gamma}$ with $\gamma = 0.1$, the set 
\begin{equation} \label{eq:C:Ginzburg:Landau}
  \msc = \left\{ x \in \rset^d : -\log(\pi(x)) - \mathrm{Cst}_{\mathrm{GL}}>100 \right\}
\end{equation}
and the teleportation kernel $Q = \Qrwm_{\upsigma}$ with $\upsigma = 0.1$. Note that the step size $\gamma$ used in the MALA is set to be quite small in comparison to that used in the KKT sampler; otherwise, the acceptance probability would degenerate to zero, resulting in a stuck algorithm. This is mainly due to the fact that $\nabla \log \pi$ is not Lipschitz in this case. In contrast, the teleportation algorithm does not suffer from this issue if the set $\Cset$ is tuned appropriately. Indeed, with the choice \eqref{eq:C:Ginzburg:Landau} of $\Cset$, the MALA is used only in the neighborhoods of the modes of $\log \pi$, where $\nabla \log \pi$ is Lipschitz with a relatively small Lipschitz constant, while instead RWM---incorporating no gradient information---is used outside these regions. 

In \Cref{fig:Ginzburg}, we display the  ESSs 
of the trajectories produced by the two methods, divided by the number of evaluations of either $\log \pi$ or $\nabla \log \pi$ per iteration. The values are calculated on the basis of $10^5$ iterations (after discarding burn-in periods comprising $10^5$ iterations). Evidently, the KKT sampler outperforms clearly MALA, with an improvement of at least an order of magnitude in terms of average ESS.

\begin{table}[]
\begin{tabular}{|l|c|c|c|c|}
  \hline
  Algorithm & mean & variance & min & max\\
  \hline
  KKT sampler & 908 & 5438 & 727 & 1091\\
    \hline
  MALA & 34 & 83 & 12 & 57\\
  \hline
\end{tabular}
\caption{Summary statistics of ESSs calculated on the basis of MCMC paths generated by the MALA and the KKT sampler for the Ginzburg--Landau model in \Cref{sec:Ginzburg:Landau:model}. The values are based on $10^5$ iterations of each algorithm.}
\label{fig:Ginzburg}
\end{table}


\section*{Acknowledgments}

The work of J.~Olsson is supported by the Swedish Research Council, Grant~2018-05230. 

\appendix

\section{Kernel notation}
\label{sec:notation}

The following kernel notation will be used at several places in the paper. 
Let $(\Eset_1, \Esigma_1)$ and $(\Eset_2, \Esigma_2)$ be general measurable spaces. A possibly unnormalized transition kernel $K$ on $\Eset_1 \times \Esigma_2$ induces the following three operations, one on $\bddfunc{\Esigma_2}$ and two on $\measureset(\Esigma_1)$: 
\begin{itemize}
\item For $h \in \bddfunc{\Esigma_2}$ we define the function $Kh : \Eset_1 \ni x \mapsto \int h(y) \, K(x, \rmd y)$ (whenever the  integral is well defined). 
\item For $\mu \in \measureset(\Esigma_1)$ we define the measure $\mu K: \Esigma_2 \ni \Aset \mapsto \int \mu(\rmd x) \, K(x, \Aset)$. 
\item For $\mu \in \measureset(\Esigma_1)$ we define the measure $\mu \tensprod K: \Esigma_1 \tensprod \Esigma_2 \ni \Aset \mapsto \int \! \int _\Aset\mu(\rmd x) \, K(x, \rmd y)$. 
\end{itemize}
Now let $(\Eset_3, \Esigma_3)$ be a third measurable space and $L$ a transition kernel on $\Eset_2 \times \Esigma_3$. Then we define the following products between $K$ and $L$, the first resulting in a kernel on $\Eset_1 \times \Esigma_3$ and the second in a kernel on $\Eset_1 \times (\Esigma_2 \tensprod \Esigma_3)$:
\begin{itemize}
\item $K L : \Eset_1 \times \Esigma_3 \ni (x, \Aset) \mapsto \int K(x, \rmd y) \, L(y, \Aset)$.
\item $K \tensprod L : \Eset_1 \times (\Esigma_2 \tensprod \Esigma_3) \ni (x, \Aset) \mapsto \int \! \int_\Aset K(x, \rmd y) \, L(y, \rmd z)$. 
\end{itemize}


\section{Proofs and technical lemmas}
\label{sec:appendix}
\subsection{Proof of \Cref{prop:nonrev}}
\label{sec:S:nonRev}

First, note that for all $\Aset_0$ and $\Aset_1$ in $\Esigma$, by the definition of $S$,
\begin{align} \nonumber 
    &[\pi \tensprod S](\Aset_0 \times \Aset_1) \\
    &\qquad\qquad=\int_{\Aset_0} \, \pi(\rmd y) P(y, \Cset^c \cap \Aset_1)+\int_{\Aset_0} \pi(\rmd y) \, P(y,\Cset) \pi_\Cset(\Aset_1) \nonumber \\
    &\qquad\qquad= [\pi \tensprod P]((\Cset^c \cap \Aset_0) \times (\Cset^c \cap \Aset_1))+[\pi \tensprod P]((\Cset \cap \Aset_0) \times (\Cset^c \cap \Aset_1)) \nonumber \\      
    &\qquad\qquad\qquad + [\pi \tensprod P]((\Cset^c \cap \Aset_0) \times \Cset) \pi_{\Cset}(\Aset_1) + [\pi \tensprod P]((\Cset \cap \Aset_0) \times \Cset) \pi_{\Cset}(\Aset_1). \label{eq:proof_S_reversible_1}
\end{align}
We first prove that \ref{item:nonrev:a} and \ref{item:nonrev:b} jointly imply that $S$ is $\pi$-reversible. To do so, we rewrite each term on the right-hand side of \eqref{eq:proof_S_reversible_1} to show that $[\pi \tensprod S](\Aset_0 \times \Aset_1) = [\pi \tensprod S](\Aset_1 \times \Aset_0)$. We start with the second and the last terms. By \ref{item:nonrev:b}, 
\begin{equation} \label{eq:proof_S_reversible_3}
   [\pi \tensprod P]((\Cset \cap \Aset_0) \times (\Cset^c \cap \Aset_1)) = \pi_{\Cset}(\Aset_0) \pi(\Cset) \mu(\Cset^c \cap \Aset_1)\eqsp,
\end{equation}
and 
\begin{equation} \label{eq:proof_S_reversible_4}
    [\pi \tensprod P]((\Cset \cap \Aset_0) \times \Cset) = \int_{\Cset \cap \Aset_0} \pi(\rmd y) \, (1 - P(y, \Cset^c)) = (1 - \mu(\Cset^c)) \pi(\Cset) \pi_{\Cset}(\Aset_0) \eqsp.
\end{equation}
We turn to the third term on the right-hand side of \eqref{eq:proof_S_reversible_1}. Under \ref{ass:P:inv}, \ref{item:nonrev:a} and \ref{item:nonrev:b} imply
\begin{align} \nonumber
    [\pi \tensprod P]((\Cset^c\cap\Aset_0) \times \Cset)&=\pi(\Cset^c\cap\Aset_0)-[\pi \tensprod P]((\Cset^c\cap\Aset_0) \times \Cset^c) \nonumber \\
    &= [\pi\tensprod P](\Eset \times (\Cset^c\cap\Aset_0)) - [\pi \tensprod P](\Cset^c\times(\Cset^c\cap\Aset_0)) \nonumber \\
    &= [\pi \tensprod P](\Cset\times(\Cset^c\cap\Aset_0)) \nonumber \\
    &= \pi(\Cset )\mu(\Cset^c\cap\Aset_0)\eqsp. \label{eq:proof_S_reversible_2}
\end{align}
Combining \eqref{eq:proof_S_reversible_1}--\eqref{eq:proof_S_reversible_2} yields 
\begin{align*}
    [\pi \tensprod S](\Aset_0 \times \Aset_1)&=[\pi\tensprod P]((\Cset^c\cap \Aset_0) \times (\Cset^c \cap \Aset_1))+\pi_{\Cset}(\Aset_0)\pi(\Cset)\mu(\Cset^c \cap \Aset_1) \\
    &\quad + \pi_{\Cset}(\Aset_1)\pi(\Cset)\mu(\Cset^c\cap\Aset_0)+(1-\mu(\Cset^c))\pi(\Cset)\pi_{\Cset}(\Aset_0)\pi_{\Cset}(\Aset_1)\eqsp, 
\end{align*}
and by \ref{item:nonrev:a} this expression is symmetric in $\Aset_0$ and $\Aset_1$. Thus, $[\pi \tensprod S](\Aset_0 \times \Aset_1) = [\pi \tensprod S](\Aset_1 \times \Aset_0)$, which implies that $S$ is $\pi$-reversible.
  
We now establish the converse. If $S$ is $\pi$-reversible, then \ref{item:nonrev:a} holds trivially true (as $\mce_{\Cset^c}^{\tensprod 2} \subset \mce^{\tensprod 2}$). We show  that \ref{item:nonrev:b} holds as well. By \eqref{eq:proof_S_reversible_1}, and using that $S$ is $\pi$-reversible, we have for every $(\Aset, \Bset) \in \Esigma_{\Cset} \times \Esigma_{\Cset^c}$,
\begin{equation*}
    \int_{\Aset} \pi(\rmd y) \, P(y,\Bset) = [\pi \tensprod S](\Aset\times \Bset) = [\pi \tensprod S](\Bset\times \Aset) = \pi_{\Cset}(\Aset) \int_{\Bset} \pi(\rmd y) \, P(y,\Cset)\eqsp. 
\end{equation*}
Therefore, 
\begin{equation} \label{eq:proof_S_reversible_5}
    \int_{\Aset} \pi(\rmd y) \, (P(y,\Bset) - \mu(\Bset)) = 0\eqsp,
\end{equation}
where we have defined the measure $\mu : \Esigma_{\Cset^c} \ni \Bset' \mapsto \int_{\Bset'} \pi(\rmd y) \, P(y,\Cset) / \pi(\Cset)$. Since $\Esigma_{\Cset^c}$ is countably generated there exists a $\pi$-system $\seq{\Bset_k}{k \in \nset}$, with $\Bset_k \in \Esigma_{\Cset^c}$ for all $k$, such that $\Esigma_{\Cset^c} = \sigma(\seq{\Bset_k}{k \in \nset})$ and $\Bset_0 = \Cset^c$. Then, by using, for every $k \in \nset$, \eqref{eq:proof_S_reversible_5} with $\Bset=\Bset_k$ and $\Aset = \{P(\cdot, \Bset_k) > \mu(\Bset_k)\}$ as well as $\Aset = \{P(\cdot, \Bset_k) < \mu(\Bset_k)\}$, we conclude that there exists $\msx_k \in \mce_\Cset$ such that $\pi(\msx_k^c)=0$ and for every $y \in \msx_k$, $P(y,\Bset_k) = \mu(\Bset_k)$. Now, let $\msx \eqdef \cap_{k \in\nset} \msx_k$; then $\pi(\msx^c)=0$ and for every $y \in \msx$ and $k \in \nset$, $P(y,\Bset_k) = \mu(\Bset_k)$. Therefore, by Dynkin's $\pi$--$\lambda$ theorem, $P(y,\cdot)|_{\Cset^c} = \mu$ for $\pi$-almost all $y \in \Cset$, which completes the proof.
  
\subsection{Proof of \Cref{propo:harmonic_LLN}}
\label{sec:harmonic}
 First, assume that \ref{propo:harmonic_LLN_ii} holds. Then by \cite[Theorem 5.1.8]{douc:moulines:priouret:soulier:2018}, for every measurable function $g: \Eset \to \rset$ such that $\pi(\abs{g}) < \plusinfty$,
\begin{equation} \label{eq:techn:lgn:zero}
    \lim_{n \to \plusinfty} n^{-1} \sum_{k = 0}^{n - 1} g(X_k) = \PE^P_\pi[g(X_0) \mid \invar], \quad \PP^P_{\pi}\mbox{-a.s.},
\end{equation}
where $\invar \eqdef \set{\Aset \in \Esigma^{\tensprod\nset}}{\1_{\Aset}=\1_{\Aset}\circ \shift}$ is the $\sigma$-field of invariant sets. Under \ref{propo:harmonic_LLN_ii}, \cite[Corollary 5.2.4]{douc:moulines:priouret:soulier:2018} implies that the invariant random variable $\PE_\pi^P[g(X_0) \mid \invar]$ is $\PP_\pi^P$-a.s. constant. As a consequence, $\PP_\pi^P$-a.s.,
\begin{equation*}
    \PE^P_\pi[g(X_0) \mid \invar] = \PE^P_\pi \lrb{\PE^P_\pi[g(X_0) \mid \invar]} = \PE^P_\pi[g(X_0)] =\pi(g). 
\end{equation*}
Plugging this into \eqref{eq:techn:lgn:zero} yields $\PP^P_\pi(\Aset)=1$, where 
\begin{equation*}
    \Aset \eqdef \defEns{\lim_{n\to \plusinfty} n^{-1} \sum_{k=0}^{n-1} g(X_k) = \pi(g)}.
\end{equation*}
Since $\Aset$ is invariant, \ie, $\1_{\msa} = \1_{\msa} \circ \theta$, the function $h: \Eset \ni x\mapsto \PP_x^P(\Aset)$ is bounded and harmonic, and by \ref{propo:harmonic_LLN_ii} it is hence equal to a constant $\zeta$. Then, $1=\PP^P_\pi(\Aset)=\int \pi(\rmd x) \, \PP_x^P(\Aset)=\zeta$ and we conclude that $\PP_\xi^P(\Aset)=\int \xi(\rmd x) \, \PP_x^P(\Aset)=\zeta=1$. Thus, \ref{propo:harmonic_LLN_i} holds true. 
    
Conversely, assume \ref{propo:harmonic_LLN_i} and let $h$ be a bounded and harmonic function.  Applying \ref{propo:harmonic_LLN_i} with $g=h$ yields that for every $x \in \Eset$,
\begin{equation*}
    \lim_{n\to \plusinfty}  n^{-1} \sum_{k=0}^{n-1} h(X_k) = \pi(h), \quad \PP^P_{x}\mbox{-a.s.} 
\end{equation*}
In addition, since by \cite[Proposition~5.2.2(ii)]{douc:moulines:priouret:soulier:2018}, for any $x \in\mse$, $\seq{h(X_k)}{k \in \nset}$ converges $\PP^P_x$-a.s. to some random variable $Y=\limsup_{k \to \plusinfty} h(X_k)$ such that $h(x) = \PE^P_x[Y]$, we finally conclude that $Y=\pi(h)$, $\PP^P_{x}$-a.s., and $h(x) = \PE^{P}_x[Y] = \pi(h)$ for any $x \in \Eset$. This shows \ref{propo:harmonic_LLN_ii}.  
    
\subsection{Technical results for proving \Cref{thm:LLN_Markov} and \Cref{thm:ergo:geom}}

\begin{lemma} \label{lem:tech}
Set $\bCset=\Cset \times \Cset$. Then for every $(y, z) \in  \Eset \times \Cset$, $n \in \nset$ and $v \in \posfunc{\Esigma}$,
\begin{equation} \label{eq:induction}
\PEE{R}{(y, z)}[\indin{\rti_{\bCset} > n} v(Y_n)] = \PEE{P}{y}[\indin{\rti_{\Cset} > n} v(X_n)]\eqsp.
\end{equation}
If, in addition, $\inf_{y \in\mse}\PPP{P}{y}(\rti_\Cset<\plusinfty)=1$, then the following properties hold true. 
\begin{enumerate}[(i)]
\item \label{item:tech:one} For every $(y, z) \in \Eset \times \Cset$ and $\ell \in \nsetpos$,
$
\PPP{R}{(y,z)}(\rti^\ell_{\bCset}<\plusinfty)=1.
$
\item \label{item:tech:two} For every $(g, h) \in \posfunc{\nset} \times \posfunc{\Csigma}$ and $\Xinit \in \measureset_1(\Esigma \tensprod \Csigma)$,
\begin{equation}
\PEE{R}{\Xinit}[g(\rti_\bCset)h(Z_{\rti_\bCset})] = \int
\Xinit(\rmd (y, z)) \, \PEE{P}{y}[g(\rti_{ \Cset})] Qh(z) \label{eq:tech:one} 
\end{equation}
and for every $\ell \in \nsetpos$,
\begin{equation}
\PEE{R}{\Xinit}\left[ g(\rti_\bCset \circ{\shift_{\rti^{\ell-1}_\bCset}})h(Z_{\rti^{\ell}_\bCset}) \mid \mcf_{\rti^{\ell-1}_{\bCset}} \right] = \PEE{P}{Y_{\rti^{\ell-1}_\Cset}}[g(\rti_{\Cset})]Qh(Z_{\rti^{\ell-1}_{\bCset}})\eqsp, \quad \PPP{R}{\Xinit}\mbox{-a.s.}\label{eq:tech:two}
\end{equation}
\item \label{item:tech:three} For all $k \in \nsetpos$, $(g_i)_{i=1}^{k}$ in $\posfunc{\nset}$, $(h_i)_{i=1}^{k}$ in  $\posfunc{\Csigma}$ and $\Xinit \in \measureset_1(\Esigma \tensprod\Csigma)$,
\begin{multline} \label{eq:tech:three}
\PEE{R}{\Xinit}\lrb{\prod_{\ell=1}^k g_\ell(\rti^\ell_\bCset-\rti^{\ell-1}_\bCset ) h_\ell(Z_{\rti^\ell_\bCset})} \\
= \int \Xinit(\rmd (y, z)) \, \PEE{P}{y}[g_1(\rti_{\Cset})]\PEE{Q}{z}\lrb{h_k(X_{k})\prod_{\ell=1}^{k-1} h_\ell(X_{\ell}) \PEE{P}{X_\ell}[g_{\ell+1}(\rti_\Cset)]}\eqsp,
\end{multline}
where, by convention, $\rti^0_\bCset \eqdef 0$.
\item \label{eq:tech:four} Let $(y, z) \in \Eset\times \Cset$; then under $\PP_{(y, z)}^R$, $\seq{Y_{\rti_\bCset^k}}{k \in \nsetpos}$  is a Markov chain with transition kernel $Q$ and initial distribution $\delta_z$. 
\end{enumerate}
\end{lemma}

\begin{proof}
We establish \eqref{eq:induction} by induction over $n$. The base case $n = 0$ holds trivially true. Assuming that \eqref{eq:induction} holds true for some $n \in \nset$, the Markov property implies that for every $(y,z) \in\mse\times \msc$,
\begin{align*}
    \PEE{R}{(y,z)}[\indin{\rti_{\bCset} > n+1} v(Y_{n+1})]&= \PEE{R}{(y,z)}[\indin{\rti_{\bCset}>n} \indin{Y_{n+1} \notin \Cset} v(Y_{n+1})] \\
    &= \PEE{R}{(y,z)}[\indin{\rti_{\bCset} > n} P(\indi{\Cset^c}v)(Y_{n})] \\
    &= \PEE{P}{y}[\indin{\rti_{\Cset} > n} P(\indi{\Cset^c}v)(X_{n})] \\ 
    &= \PEE{P}{y}[\indin{\rti_{\Cset} > n+1} v(X_{n+1})]\eqsp,
\end{align*}
which establishes \eqref{eq:induction} with $n$ replaced by $n+1$ and hence completes the induction step. 

We now establish (i)--(iv) in turn.  

\begin{enumerate}[(i)]
\item By setting $v \equiv 1$ in \eqref{eq:induction} and letting $n\to \infty$, we get that for every $(y, z) \in \Eset \times \Cset$, $\PPP{R}{(y,z)}(\rti_{\bCset} = \plusinfty) = \PPP{P}{y}(\rti_{\Cset} =\plusinfty)$. Therefore, by assumption, for every $(y,z)\in\mse\times \msc$,
$$
    \PPP{R}{(y,z)}(\rti_{\bCset}<\plusinfty)=\PPP{P}{y}(\rti_{\Cset}<\plusinfty)=1\eqsp,
$$
from which it follows that $\PPP{R}{(y,z)}(\rti^\ell_{\bCset}<\plusinfty) = 1$ for every $\ell \in \nsetpos$. 
\item We first show \eqref{eq:tech:one}. Using the Markov property and the definition \eqref{eq:def_R_Kac} of $R$, for every $(y,z) \in \Eset \times \Cset$,
\begin{align*}
    \PEE{R}{(y,z)}[\indin{\rti_\bCset=n} h(Z_n)]&=\PEE{R}{(y,z)}[\indin{\rti_\bCset>n-1}\indin{Y_n \in \Cset} h(Z_n)]\\
    &= \PEE{R}{(y,z)}[\indin{\rti_\bCset > n-1}P(Y_{n-1},\Cset) Qh(Z_{n-1})]\\
    &= \PEE{R}{(y,z)}[\indin{\rti_\bCset > n-1}P(Y_{n-1},\Cset)] Qh(z)\\
    &=\PEE{P}{y}[\indin{\rti_\Cset > n-1}P(X_{n-1},\Cset)] Qh(z) \\
    &=\PPP{P}{y}(\rti_\Cset=n) Qh(z)\eqsp,
\end{align*}
where the penultimate equality follows from \eqref{eq:induction}. Then, using \ref{item:tech:one} and the previous identity, 
\begin{align*}
 \PEE{R}{(y,z)}[g(\rti_\bCset)h(Z_{\rti_\bCset})] &=\sum_{n=1}^\plusinfty g(n) \PEE{R}{(y,z)}[\indin{\rti_\bCset=n} h(Z_n)] \\
 & =\sum_{n=1}^\plusinfty g(n) \PPP{P}{y}(\rti_\Cset=n) Qh(z) \\
 &= \PEE{P}{y}[g(\rti_{\Cset})]Qh(z)\eqsp.
\end{align*}
Integrating with respect to  $\Xinit$ establishes \eqref{eq:tech:one}. 

We turn to \eqref{eq:tech:two}. Since $Z_{\rti^{\ell}_{\bCset}} = Z_{\rti_\bCset} \circ \shift_{\rti^{\ell-1}_{\bCset}}$ on $\{\rti^{\ell-1}_{\bCset}<\plusinfty\}$, combining the strong Markov property, \ref{item:tech:one} and \eqref{eq:tech:one} yields, $\PPP{R}{\Xinit}$-\as,
\begin{align*}
\PEE{R}{\Xinit} \left[ g(\rti_\bCset \circ{\shift_{\rti^{\ell-1}_{\bCset}}}) h(Z_{\rti^{\ell}_{\bCset}}) \mid \mcf_{\rti^{\ell-1}_{\bCset}} \right] &=
\PEE{R}{\Xinit} \left[ g(\rti_\bCset \circ{\shift_{\rti^{\ell-1}_{\bCset}}})h(Z_{\rti_\bCset} \circ \shift_{\rti^{\ell-1}_{\bCset}}) \mid \mcf_{\rti^{\ell-1}_{\bCset}} \right] \\
&= \PEE{R}{(Y_{\rti^{\ell-1}_{\bCset}},Z_{\rti^{\ell-1}_{\bCset}})}[g(\rti_\bCset) h(Z_{\rti_\bCset})]\\
&= \PEE{P}{Y_{\rti^{\ell-1}_{\bCset}}}[g(\rti_\Cset)]Qh(Z_{\rti^{\ell-1}_{\bCset}})\eqsp,
\end{align*}
where the last equality follows from  \eqref{eq:tech:one}. This completes the proof \eqref{eq:tech:two}.
\item We proceed by induction. The base case $k=1$ holds by \eqref{eq:tech:one}. In order to carry through the induction step, assume that \eqref{eq:tech:three} holds for some $k \in \nsetpos$. Then, for all $(g_i)_{i=1}^{k+1}$ in $\posfunc{\nset}$,  $(h_i)_{i=1}^{k+1}$ in $\posfunc{\Csigma}$ and $\Xinit \in \measureset_{1}(\Esigma \tensprod\Csigma)$, using the tower property, \eqref{eq:tech:two} and the fact that $Y_{\rti^{\ell}_{\bCset}}=Z_{\rti^{\ell}_{\bCset}}$, $\PPP{R}{\Xinit}$-a.s., for all $\ell \in \nsetpos$, we obtain
\begin{align*}
     \lefteqn{\PEE{R}{\Xinit}\lrb{\prod_{\ell=1}^{k+1} g_\ell(\rti^\ell_\bCset-\rti^{\ell-1}_\bCset ) h_\ell(Z_{\rti^\ell_\bCset})}}\\
     &= \PEE{R}{\Xinit}\lrb{\PEE{R}{\Xinit}[g_{k+1}(\rti_\bCset \circ{\shift_{\rti^{k}_\bCset}})h_{k+1}(Z_{\rti^{k+1}_\bCset}) \mid \mcf_{\rti^{k}_{\bCset}}]\prod_{\ell=1}^k g_\ell(\rti^\ell_\bCset-\rti^{\ell-1}_\bCset ) h_\ell(Z_{\rti^\ell_\bCset})} \nonumber \\
&= \PEE{R}{\Xinit}\lrb{\PEE{P}{Z_{\rti^{k}_\bCset}}[g_{k+1}(\rti_{\Cset})]Qh_{k+1}(Z_{\rti^{k}_{\bCset}})\prod_{\ell=1}^k g_\ell(\rti^\ell_\bCset-\rti^{\ell-1}_\bCset ) h_\ell(Z_{\rti^\ell_\bCset})}\eqsp. 
\end{align*}
The expectation on the right-hand side is over a product of nonnegative functions of $Z_{\rti^{\ell}_\bCset}$ and $\rti^\ell_\bCset-\rti^{\ell-1}_\bCset$ for $\ell \in \{1,\ldots,k\}$. Therefore, using the induction hypothesis,
\begin{align*}
    \lefteqn{\PEE{R}{\Xinit}\lrb{\prod_{\ell=1}^{k+1} g_\ell(\rti^\ell_\bCset-\rti^{\ell-1}_\bCset ) h_\ell(Z_{\rti^\ell_\bCset})}} \\
    &= \int \Xinit(\rmd (y, z)) \, \PEE{P}{y}[g_1(\rti_{\Cset})] \\ 
    &\hspace{16mm} \times \PEE{Q}{z} \lrb{\PEE{P}{X_{k}}[g_{k+1}(\rti_{\Cset})]Qh_{k+1}(X_k) h_k(X_k) \prod_{\ell=1}^{k-1} h_\ell(X_{\ell}) \PEE{P}{X_\ell}[g_{\ell+1}(\rti_\Cset)]}\\
    &= \int \Xinit(\rmd (y, z)) \, \PEE{P}{y}[g_1(\rti_{\Cset})] \PEE{Q}{z} \lrb{h_{k+1}(X_{k+1})\prod_{\ell=1}^k h_\ell(X_{\ell}) \PEE{P}{X_\ell}[g_{\ell+1}(\rti_\Cset)]}\eqsp,  
\end{align*}
which means that \eqref{eq:tech:three} holds true with $k$ replaced by $k + 1$. Thus, the claim follows by induction. 
\item By applying \eqref{eq:tech:three} with $g_1 \equiv \ldots \equiv g_k \equiv \mathbf{1}$, we conclude that for all $k \in \nsetpos$, $(h_i)_{i=1}^{k+1}$ in $\posfunc{\Csigma}$ and $\Xinit \in \measureset_{1}(\Esigma \tensprod\Csigma)$,
\begin{equation*}
  \PEE{R}{\Xinit} \lrb{\prod_{\ell=1}^k h_\ell(Z_{\rti^\ell_\bCset})} = \int \Xinit(\rmd (y, z)) \, \PEE{Q}{z}\lrb{\prod_{\ell=1}^{k} h_\ell(X_{\ell})}\eqsp. 
\end{equation*}
This shows that for every $(y,z)\in\mse\times \msc$, under $\PP^R_{(y,z)}$, $\seq{Z_{\rti^\ell_\Cset}}{\ell \in \nsetpos}$ is a Markov chain with transition kernel $Q$ and initial distribution $\delta_{(y,z)}$. We may now complete the proof by noting that $Y_{\rti^{\ell}_{\bCset}} = Z_{\rti^{\ell}_{\bCset}}$, $\PPP{R}{\Xinit}$-a.s., for all $\ell \in \nsetpos$.
\end{enumerate}
\end{proof}

\begin{lemma}  \label{lem:smallSet}
    Assume \ref{ass:Q:geom}\ref{ass:geom:D}--\ref{ass:geom:mino:C} and that for every $y \in \Eset$, $\PPP{P}{y}(\rti_\Cset<\plusinfty)=1$. 
    Set $\delta \eqdef \inf_{z \in \Dset} P(z, \Cset)$ and $\bDset \eqdef \Dset \times \Dset$ and define the probability measure $\bar{\nu} : \Esigma \tensprod \Csigma \ni \Aset \mapsto \int_\Cset \nu(\rmd z) \, \indi{\Aset}(z,z)$. Then the set $\bDset=\Dset \times \Dset$ is an accessible $(1,\delta \epsilon \bar \nu)$-small set for $R$ such that $\bar \nu(\bDset)>0$.
\end{lemma}
\begin{proof}
Under \ref{ass:Q:geom}\ref{ass:geom:D}--\ref{ass:geom:mino:C} it holds that for every $(y,z)\in \bDset$ and $\Aset \in \Esigma \tensprod \Csigma$,
$$
R((y,z), \Aset) \geq P(y,\Cset) \int Q(z,\rmd z') \indi{\Aset}(z',z') \geq \delta \epsilon \bar \nu(\Aset)\eqsp,
$$
showing that $\bDset$ is a $(1,\delta \epsilon\bar\nu)$-small set for $R$. In addition, $\bar \nu(\bDset)=\nu(\Dset)>0$. To complete the proof, it remains to show that $\bar{\msd}$ is accessible for $R$. By assumption, $Q$ admits an accessible small set and is hence irreducible. Furthermore, \cite[Theorem 9.2.4]{douc:moulines:priouret:soulier:2018} implies that there exists a maximal irreducibility measure $\psi \in \measureset(\Csigma)$ for $Q$. We define the measure $\bar{\psi} : \Esigma \tensprod \Csigma \ni \Aset \mapsto \int_\Cset \psi(\rmd z) \, \indi{\Aset}(z,z)$ and complete the proof by showing that $R$ is $\bar{\psi}$-irreducible. Indeed, since $\nu$ is, by \ref{ass:Q:geom}\ref{ass:geom:D}, an irreducibility measure for $Q$ such that $\nu(\msd) > 0$, it holds that $\bar{\psi}(\bar{\msd}) = \psi(\msd) >0$; thus, by \cite[Lemma 3.5.2]{douc:moulines:priouret:soulier:2018}, $\bar{\msd}$ is accessible for $R$ if $R$ is $\bar \psi$-irreducible. 

To establish $\bar{\psi}$-irreducibility, consider a set $\Aset \in \Esigma \tensprod \Csigma$ such that $\bar{\psi}(\Aset) > 0$. Setting $\Aset_1 \eqdef \set{z \in \Cset}{(z,z) \in \Aset}$, it holds that $\bar \psi(\Aset) = \int \1_{\Aset}(z,z) \, \psi(\rmd z) =\psi(\Aset_1)>0$, and since $Q$ is $\psi$-irreducible, it follows that $\sum_{k=1}^{\plusinfty} Q^k(z,\Aset_1)>0$ for all $z \in \Cset$. Now for all $(y, z) \in \Eset \times \Cset$, $Y_k = Z_k$, $\PPP{R}{(y, z)}$-a.s. on $\{Y_k \in \Cset\}$, and therefore
\begin{align*}
\sum_{n=1}^{\plusinfty} R^n((y,z),\Aset) &\geq \sum_{n=1}^{\plusinfty} \PPP{R}{(y,z)}\lr{(Y_n, Z_n) \in \Aset \cap \bCset, Y_n = Z_n} \\
&= \sum_{\ell=1}^{\plusinfty} \PPP{R}{(y,z)}\lr{Z_{\rti^\ell_\bCset} \in \Aset_1}=\sum_{\ell=1}^{\plusinfty} Q^\ell(z,\Aset_1)>0\eqsp,
\end{align*}
where the first and last equalities follow from \Cref{lem:tech}\ref{item:tech:one} and \Cref{lem:tech}\ref{item:tech:three}, respectively. The Markov kernel $R$ is therefore $\bar \psi$-irreducible, and the proof is finalized. 
\end{proof}


\bibliographystyle{imsart-number} 
\bibliography{../../../Bibliography/biblio}       


\end{document}